\documentclass[aps, twocolumn, notitlepage, nofootinbib, superscriptaddress, prb]{revtex4-2}

\makeatletter
\AtBeginDocument{\let\LS@rot\@undefined}
\makeatother

\usepackage{subfigure}
\usepackage{longtable}
\usepackage{booktabs}
\usepackage{textcomp}
\usepackage{physics}
\usepackage{graphicx}
\usepackage{amssymb}
\usepackage{amsmath}
\usepackage{bm}
\usepackage{amsmath, amsthm, amssymb}
\usepackage{dsfont}
\usepackage[colorlinks=true,linkcolor=blue,urlcolor=blue,citecolor=blue]{hyperref}
\usepackage{multirow}
\usepackage{url}
\usepackage{bbold}
\usepackage{array}
\usepackage{color}

\newcolumntype{M}[1]{>{\centering\arraybackslash}m{#1}}
\newcolumntype{N}{@{}m{0pt}@{}}
\newcolumntype{P}[1]{>{\centering\arraybackslash}p{#1}}

\newtheorem{definition}{Definition}[section]

\newtheorem{lemma}[definition]{Lemma}
\newtheorem{corollary}[definition]{Corollary}

\newcommand{\ignore}[1]{}

\newcommand{\eq}[1]{Eq.\thinspace{}(\ref{#1})}

 \makeatletter
\newsavebox{\@brx}
\newcommand{\llangle}[1][]{\savebox{\@brx}{\(\m@th{#1\langle}\)}%
  \mathopen{\copy\@brx\kern-0.5\wd\@brx\usebox{\@brx}}}
\newcommand{\rrangle}[1][]{\savebox{\@brx}{\(\m@th{#1\rangle}\)}%
  \mathclose{\copy\@brx\kern-0.5\wd\@brx\usebox{\@brx}}}
\makeatother

\def\beq{\begin{equation}}
\def\eeq{\end{equation}}
\def\bea{\begin{eqnarray}}
\def\eea{\end{eqnarray}}

\usepackage{pdfpages} 
\makeatletter
\AtBeginDocument{\let\LS@rot\@undefined}
\makeatother

\begin{document}
\title{Many body localization transition with correlated disorder}

\author{Zhengyan Darius Shi}
\affiliation{Department of Physics, Massachusetts Institute of Technology, Cambridge, MA 02139, USA}
\author{Vedika Khemani}
\affiliation{Department of Physics, Stanford University, Stanford, CA 94305, USA}
\author{Romain Vasseur}
\affiliation{Department of Physics, University of Massachusetts, Amherst, MA 01003, USA}
\author{Sarang Gopalakrishnan}
\affiliation{Department of Physics, The Pennsylvania State University, University Park, PA 16802, USA}

\begin{abstract}

We address the critical properties of the many-body localization (MBL) phase transition in one-dimensional systems subject to spatially correlated disorder. We consider a general family of disorder models, parameterized by how strong the fluctuations of the disordered couplings are when coarse-grained over a region of size $\ell$. For uncorrelated randomness, the characteristic scale for these fluctuations is $\sqrt{\ell}$; more generally they scale as $\ell^\gamma$. We discuss both positively correlated disorder ($1/2 < \gamma < 1$) and anticorrelated, or ``hyperuniform,'' disorder ($\gamma < 1/2$). We argue that anticorrelations in the disorder are generally irrelevant at the MBL transition. Moreover, assuming the MBL transition is described by the recently developed renormalization-group scheme of Morningstar \emph{et al.} [Phys. Rev. B 102, 125134, (2020)], we argue that even positively correlated disorder leaves the critical theory unchanged, although it modifies certain properties of the many-body localized phase. 

\end{abstract}

\maketitle

\section{Introduction}\label{sec:intro}

In generic quantum many-body systems, interactions scramble local quantum information and bring subsystems towards thermal equilibrium~\cite{Deutsch_1991, Srednicki_1994,doi:10.1080/00018732.2016.1198134}. This ``thermalization'' process fails in the many body localized (MBL) phase~\cite{Nandkishore_Huse_2015,Abanin_Altman_Bloch_Serbyn_2019}. The MBL phase is best understood for systems subject to strong spatial randomness, but is also believed to occur for quasiperiodic potentials~\cite{PhysRevB.87.134202, PhysRevLett.119.075702}. Systems deep in the MBL phase possess an extensive set of emergent local integrals of motion, ``l-bits", which leads to area law entanglement in \textit{all eigenstates} as well as Poisson statistics in the nearest-neighbor energy spacing distribution \cite{Imbrie_2016,Oganesyan-Huse2007, Pal_Huse_2010,Huse_Nandkishore_Oganesyan_2014,Serbyn_Papic_Abanin_2013}. This is to be contrasted with the thermal phase that features volume-law entanglement and local energy level repulsion conforming to random matrix theory.

While the physical picture deep inside each phase is relatively well-understood, the transition between them remains mysterious. Analytically, since the transition involves singular changes in the structure of highly excited eigenstates, critical properties cannot be extracted from a conventional low energy effective field theory approach\footnote{The MBL-thermal transition is accompanied by a transition between Poisson and Wigner-Dyson local spectral statistics.  Aspects of this spectral transition can be understood from a novel effective field theory discussed for example in Ref.~\onlinecite{Garratt_Chalker_2020} and Ref.~\onlinecite{Micklitz_Altland_2017}. Whether or not these effective field theories can describe critical scalings near the transition remains an open question.}. Furthermore, the lack of exactly solvable models makes it difficult to pinpoint key ingredients that would go into a general conceptual framework. Numerically, state-of-the-art exact-diagonalization is limited to small system sizes $L \lesssim 25$ (see e.g. Refs.~\onlinecite{luitz2015many,  clarkbimodal,khemani2017critical, Panda_2020, Sierant_Lewenstein_Zakrzewski_2020}) and produces a correlation length exponent $\nu$ that strongly violates the rigorous Harris bound $\nu \geq 2$ in one dimension\cite{Harris_1974,Chayes_Chayes_Fisher_Spencer_1986, CLO}, making any attempt to extrapolate the critical scaling difficult in the near future. 

Recently, significant progress has been made by means of approximate or phenomenological real-space renormalization-group (RG) approaches~\cite{Vosk_Huse_Altman_2015,Potter_Vasseur_Parameswaran_2015,Zhang_Zhao_Devakul_Huse_2016,Dumitrescu_Vasseur_Potter_2017, thiery2017microscopically, Goremykina_Vasseur_Serbyn_2019,Morningstar_Huse_2019}.
In strongly disordered systems, one should in general consider how entire probability distributions of couplings flow as the system is coarse-grained. This kind of flow fits into the general framework of \textit{strong disorder renormalization group (SDRG)} invented by Dasgupta-Hu-Ma in Ref.~\onlinecite{Ma_Dasgupta_Hu_1979}, rigorously developed by Fisher in Ref.~\onlinecite{Fisher_1992, Fisher_1995} and subsequently applied to numerous examples (see Refs.~\onlinecite{Igloi_Monthus_2005, Igloi_Monthus_2018} for a review). On the analytic side, the main appeal of SDRG is that the flow equations sometimes admit exact solutions, giving solvable models that are unavailable at the microscopic level; on the numerical side, the computational complexity of SDRGs is exponentially lower than that of exact diagonalization, allowing simulations with $\mathcal{O}(10^7)$ initial degrees of freedom and thus providing a better chance of accessing the critical scaling.

Note that these RGs only attempt to describe the {\em asymptotic} MBL transition at the largest length and timescales, while the finite-size or finite-time MBL crossover observed in numerics (see Ref.~\onlinecite{2021arXiv211111455S} for recent results with correlated disorder) and experiments is believed to be described by ``many-body resonances" involving rare superpositions of localized localized states that differ substantially in extensively many local regions ~\cite{crowley2021constructive,morningstar2021avalanches,PhysRevB.104.184203}. We will also assume the existence of the MBL phase~\cite{Basko_Aleiner_Altshuler_2006,Imbrie_2016}; see {\it e.g.} Refs.~\onlinecite{PhysRevE.102.062144,PhysRevE.104.054105,ABANIN2021168415,Panda_2020} for recent discussions.

Roughly, the existing RG schemes fall into two types. The first type starts with microscopic l-bits in the MBL phase which can delocalize due to rare resonances mediated by interactions~\cite{Potter_Vasseur_Parameswaran_2015,Dumitrescu_Vasseur_Potter_2017, thiery2017microscopically}. At strong disorder, resonant clusters are isolated and localization is robust. But as the disorder is reduced, resonances proliferate and tend to span the entire system, leading to thermalization. 

Unfortunately, due to the complexity of the cluster formation rules, no analytic solution has been possible within the first type of RGs, thus preventing a complete understanding of the critical scaling. For this reason, we focus instead on the second type of RG scheme which aims for analytic tractability at the expense of further coarse-graining. The basic strategy is to forget about individual spins and regard the system as composed of alternating thermal (T) and insulating (I) blocks initially decoupled from each other~\cite{Vosk_Huse_Altman_2015}. Disorder in the microscopic couplings then translates to disorder in a few important parameters that characterize each block---the physical length, the localization length of any putative l-bits it contains, the rate of entanglement growth within each block etc.~\cite{Vosk_Huse_Altman_2015}. When interactions between blocks are turned on, individual blocks merge into larger and larger composite blocks whose phase (T or I) and parameters are determined iteratively in terms of the parameters of their constituents. 
In the simplest block RG of this kind, there is a chain of blocks indexed by $i$, where even/odd $i$'s correspond to T/I-blocks respectively. Each block is characterized only by its physical length $l^{T/I}_i$ and at each RG step, the shortest block gets absorbed into the surrounding blocks following a ``symmetric" RG rule $l^{T/I}_{\rm new} = l^{T/I}_{i-1} + l^{I/T}_i + l^{T/I}_{i+1}$\cite{Zhang_Zhao_Devakul_Huse_2016}. While completely solvable, this RG missed the important asymmetry between T and I-blocks implied by the avalanche mechanism (a detailed review of avalanche will be given in Sec.~\ref{subsec:avalanche})\cite{DeRoeck_Huveneers_2017}. A followup paper incorporated this asymmetry and obtained a family of RGs, $l^{T/I}_{\rm new} = l^{T/I}_{i-1} + \alpha_{I/T} l^{I/T}_i + l^{T/I}_{i+1}$, controlled by the parameter $\alpha_{I/T}$\cite{Goremykina_Vasseur_Serbyn_2019}. In the maximally asymmetric limit where $\alpha_I \rightarrow \infty$ and $\alpha_T \rightarrow 0$ (henceforth referred to as GVS), the important RG directions form a two-dimensional subspace in which a critical curve separates the thermal and insulating phases. The correlation length $\xi$ is related to the distance $\delta$ from the fixed line via $\xi \sim e^{\frac{1}{\sqrt{\delta}}}$, putting GVS in the Kosterlitz-Thouless (KT) universality class. Some recent numerical studies also seem consistent with KT scaling~\cite{PhysRevB.99.134205,PhysRevResearch.2.042033,PhysRevB.102.064207,2021arXiv210609036R}. (Some previous attempts to numerically extract the correlation length  $\xi \sim \delta^{-\nu}$ from the RGS had found $\nu \approx 3.3$ as opposed to the KT value $\nu = \infty$, but these numerical values of $\nu$ exhibited considerable finite-size drifts.)

A modification of the GVS rules was proposed, and motivated on semi-microscopic grounds, in a paper by Morningstar and Huse~\cite{Morningstar_Huse_2019}. In effect, this RG scheme promoted the parameter $\alpha_I$ in GVS to a dynamical variable that has its own RG flow. The flow of the anisotropy is motivated by the following physical picture, which we will explore in more depth in Sec.~\ref{subsec:avalanche}. Slightly on the insulating side of the MBL transition, a single small thermal block can thermalize a large insulating region (of size set by the localization length) before its thermalization is eventually blocked by the discreteness of energy levels in the insulator. (At the critical point, this avalanche instead spreads throughout the system.) Thus, a typical large insulating block contains large thermal regions, through which correlations can spread without exponential suppression (i.e., that act as local short-circuits for information). These short-circuits renormalize the effective localization length, making it possible for a single thermal block to thermalize an even larger insulating region, and so on. Thus the anisotropy parameter $\alpha_I$ diverges in a specific way at the transition. The critical behavior predicted by this RG scheme was subsequently solved by Morningstar, Huse, and Imbrie~\cite{Morningstar_Huse_Imbrie_2020}, and we will refer to it as the MHI scheme in what follows. The critical exponent $\nu = \infty$ of MHI agrees with that of GVS, but the precise correlation length scaling $\xi \sim \delta^{- \log \log \delta^{-1}}$ differs from KT scaling. This can be traced to the non-analytic scale-dependence of the coefficients in the two-parameter MHI flow. From a general RG perspective this scale-dependence seems unnatural, but it has a natural physical origin in terms of the flowing anisotropy parameter. Moreover, the MHI solution lacks certain peculiar features of the GVS solution: for example, the MHI solution predicts that thermal regions in the insulating phase have a finite fractal dimension that vanishes at the transition (consistent with rare-region counting arguments), whereas the GVS solution predicts fractal dimension zero in the MBL phase.

Since the physics of the MBL transition, within the MHI scheme, is dominated by rare regions, it is natural to ask how sensitive its unusual critical properties are to assumptions about the statistics of these rare regions. A drastic way to modify these statistics is to replace the random couplings with quasiperiodically modulated couplings. In quasiperiodic systems, rare regions (to the extent that they exist) are strongly spatially correlated, apparently invalidating many of the assumptions of MHI. Indeed, some numerical studies~\cite{PhysRevB.87.134202,PhysRevLett.119.075702,PhysRevB.96.104205,PhysRevB.96.075146,Doggen_Mirlin_2019,Znidaric4595,PhysRevB.100.085105,PhysRevB.103.L220201,PhysRevB.101.035148,PhysRevB.104.214201} and an application of the GVS RG scheme suggest that the quasiperiodic MBL transition has a very different character from the random transition~\cite{Agrawal_Gopalakrishnan_Vasseur_2020_QP} (see also Ref.~\onlinecite{Zhang_Yao_2018} for a different prediction using RG approaches). If in fact there are two different universality classes of the MBL transition~\cite{PhysRevLett.119.075702}, it is natural to ask whether there might be many more, corresponding to modulations that are neither conventionally random nor quasiperiodic.

Motivated by this, we consider the MBL transition in systems subject to randomness with long-range (i.e., power-law) spatial correlations (a particular class of long-range speckle disorder has been studied numerically in Ref.~\onlinecite{Speckle}, giving results consistent with our general arguments). The interplay between long range correlations and critical exponents has a long history. For perturbations around a clean critical point, a simple scaling argument that generalizes the Harris bound~\cite{Harris_1974,Chayes_Chayes_Fisher_Spencer_1986,Chandran_Laumann_Oganesyan_2015} can give a definitive stability criterion. Suppose the correlation length scales as $\xi \sim \delta^{-\nu}$ where $\delta$ is the deviation of the order parameter from its critical value. In the presence of disorder, $\delta$ is no longer well-defined globally. Any region of linear size L can be described by the average order parameter $\bar \delta = \frac{1}{L^d} \sum_{i=1}^{L^d} \delta_i$ and the standard deviation $\sigma(\bar \delta)$. In order for the correlation length scaling of the clean critical point to be stable, the fluctuation $\sigma(\bar \delta)$ over a correlation volume $\xi^d$ must be much smaller than the mean $\delta$. For uncorrelated/short-range correlated disorder, $\sigma(\bar \delta) \sim L^{-d/2}$ by the central limit theorem. Long range correlations generally modify this scaling to $\sigma(\bar \delta) \sim L^{d(w-1)}$ ($w \in [0,1]$ is defined as the \textit{wandering exponent}), implying a simple stability criterion $\xi^{d(w-1)} = \delta^{-d(w-1)\nu} < \delta$ or equivalently $\nu > \frac{1}{d(1-w)}$. This is the \textit{correlated Harris bound}~\cite{Luck_1993} which reduces to the usual Harris bound once we take $w = 1/2$. 

However, the above bound is \emph{inadequate for inherently disordered fixed points} (the MHI fixed point being an example to keep in mind). Historically, two approaches have been taken to cure this deficiency. The first approach is perturbative and only covers stability around weak uncorrelated random fixed points. The basic strategy is to perturb a clean fixed point by weak uncorrelated disorder within a replica path integral description and run the Wilsonian RG. When the disorder is weakly relevant, the theory flows to a weak uncorrelated random fixed point. After that, one adds weak correlated disorder to the uncorrelated fixed point and run the Wilsonian RG again (see Refs.~\onlinecite{Weinrib_1984,Weinrib_Halperin_1983} for a detailed derivation that goes through all the diagrammatics). The stability criterion they found agrees with the \textit{correlated Harris bound}. The second approach seeks to derive general bounds on $\nu$ for intrinsically disordered fixed points (that may or may not arise as perturbations of uncorrelated fixed points)~\cite{Chayes_Chayes_Fisher_Spencer_1986, CLO}. The first rigorous result along this line of thinking is the Chayes-Chayes-Fisher-Spencer (CCFS) bound $\nu \geq \frac{2}{d}$ proven in Ref.~\onlinecite{Chayes_Chayes_Fisher_Spencer_1986} for arbitrary uncorrelated/short-range correlated disorder (including infinite randomness fixed points). For correlated disorder, the authors of Ref.~\onlinecite{Chayes_Chayes_Fisher_Spencer_1986} conjecture a \textit{correlated CCFS bound} $\nu \geq \frac{1}{d(1-w)}$. If true, this bound would provide a necessary but not sufficient condition for stability which is weaker than the \textit{correlated Harris bound} (for example, if the uncorrelated fixed point has exponent $\nu$, then it is unstable against correlations with wandering exponent $w$ when $\nu < \frac{1}{d(1-w)}$. But the bound gives no information on stability otherwise). 
The stability criterion $\nu \geq \frac{1}{d(1-w)}$ has been checked in all SDRGs known to date (see Refs.~\onlinecite{Igloi_Monthus_2005, Igloi_Monthus_2018} for comprehensive reviews). Therefore, it is conceivable that $\nu \geq \frac{1}{d(1-w)}$ is in fact a necessary and sufficient stability criterion for arbitrary fixed points, a conjecture that we will refer to as the \textit{generalized Harris bound}. We emphasize that to our knowledge there is no convincing argument for the validity of the generalized Harris bound at strong-randomness fixed points.

With this historical background in mind, let us return to the MHI RG. Since the uncorrelated fixed point has $\nu = \infty$, the \textit{generalized Harris bound} would suggest that long range correlations with arbitrary $w \neq 1$ cannot change $\nu$ unless $w$ flows to 1 in the IR limit. In a large class of RGs including MHI, we will explicitly calculate the flow of $w$ and show that it does not approach 1 in the IR. Therefore, if we believe in the \textit{generalized Harris bound}, the correlated fixed point still has $\nu = \infty$. This discussion leaves open the possibility that the fixed point might flow to a different universality class within the $\nu = \infty$ family. As the scaling theory of Ref.~\onlinecite{Dumitrescu_Goremykina_Parameswaran_Serbyn_Vasseur_2019} shows, any microscopic RG rule with the avalanche mechanism built in must lead to KT-scaling, so long as all $\beta$-functions are analytic. MHI evades this argument by generating logarithmic singularities in its $\beta$-functions. But since there are infinitely many types of non-analyticities, it is natural to expect that correlations can induce a different type of singularity and give rise to a new universality class. Surprisingly, under some weak assumptions, we will show that this does not happen for any initial correlation with $w < 1$ ($w=1$ corresponds to the unphysical case of perfect positive correlations), implying that the MHI universality class is stable. The effect of more general perturbations to the RG rules remains an open question that we hope will be addressed in future works.  

The rest of this paper will flesh out the above arguments in more detail. In Sec.~\ref{sec:setup_summary}, we set up the MHI RG rule of Ref.~\onlinecite{Morningstar_Huse_Imbrie_2020}, explain the quantum avalanche mechanism that motivates it, and state the three main results of the MHI analysis for uncorrelated disorder along with the corresponding modifications under long range correlated disorder. In Sec.~\ref{sec:stability_hyp}, we give general arguments for the irrelevance of hyperuniform correlations in a wide class of \textit{asymptotically additive RGs} (which includes the MHI RG as a special case). For positive correlations discussed in Sec.~\ref{sec:stability_pos}, no such general argument applies. Nevertheless, using special properties of the MBL transition, we can still show all the stability results, first via an intuitive physical argument in Sec.~\ref{subsec:stability_pos1},~\ref{subsec:stability_pos2}, and then through a more rigorous analysis of the functional RG equations in Sec.~\ref{subsec:stability_pos3} (with some technical details relegated to the appendices). Finally, in Sec.~\ref{sec:discussion} we discuss the robustness of our result to changes in the phenomenological RG rule and comment on the existence of possibly relevant perturbations (for example the quasiperiodic initial conditions considered in Refs.~\onlinecite{Agrawal_Gopalakrishnan_Vasseur_2020_QPPotts,Agrawal_Gopalakrishnan_Vasseur_2020_QP}).

\section{Model Setup and Summary of Results}\label{sec:setup_summary}

\subsection{Models of Correlated Disorder}\label{subsec:correlated_disorder}

Before introducing the RG rules, we give a microscopic motivation for how spatially correlated disorder can affect the initial conditions of the RG. For concreteness, one can have in mind the paradigmatic XXZ model
\begin{equation}
    H = \sum_{i,\beta=x,y,z} J_{\beta} \sigma^{\beta}_i \sigma^{\beta}_{i+1} + \sum_i h_i \sigma^z_i \,,
\end{equation}
where $h_i$ is a set of spatially correlated random fields with variance $W$. Given a particular choice of correlation, we assume that there exists a critical disorder strength $W = W_c$ separating the thermal phase at weak disorder and the MBL phase at strong disorder. For every random realization of $\{h_i\}$, there are contiguous regions where every $|h_i|$ is smaller than $W_c$. We refer to these contiguous regions as T-blocks (T for thermal) and the regions intervening the T-blocks as I-blocks (I for insulating). In general, the coarse-grained physical properties of each T and I-block will inherit some spatial correlations from the microscopic statistics of $\{h_i\}$. Since we will be interested only in the long-distance physics near the MBL-thermal phase transition, we will directly inject spatial correlations into the coarse-grained properties, assuming that they arise from some more complicated unspecified correlations in $\{h_i\}$.

We now state more explicitly the forms of correlated disorder that will be used in this paper. Consider a space dependent field $l_x$ with $\ev{l_x} = 0$, $\ev{l_x l_y} \sim W^2 C(|x-y|)$ with $C(\cdot)$ the position space correlation function normalized so that $C(0) = 1$. It is convenient to also introduce the Fourier transform of $C(x)$ which we refer to as the \textit{correlation spectrum} $S(k) = \int e^{ikx} C(x) dx$. Throughout the analysis, we will be interested in a family of correlations with $S(k) \sim_{k \rightarrow 0} |k|^{1-2w}$ where $w \in (0,1)$ labels the wandering exponent. When $w > 1/2$, the spatial profile $C(x) \sim |x|^{2w-2}$ is a long range \textbf{positive correlation} and coherent fluctuations are enhanced. This is in contrast to \textbf{hyperuniform correlation} with $w < 1/2$, where the spatial profile $C(x) \sim - |x|^{2w-2}$ for all $x \neq 0$ indicates long-range anti-correlation. 
\begin{figure*}
    \centering
    \includegraphics[width = \textwidth]{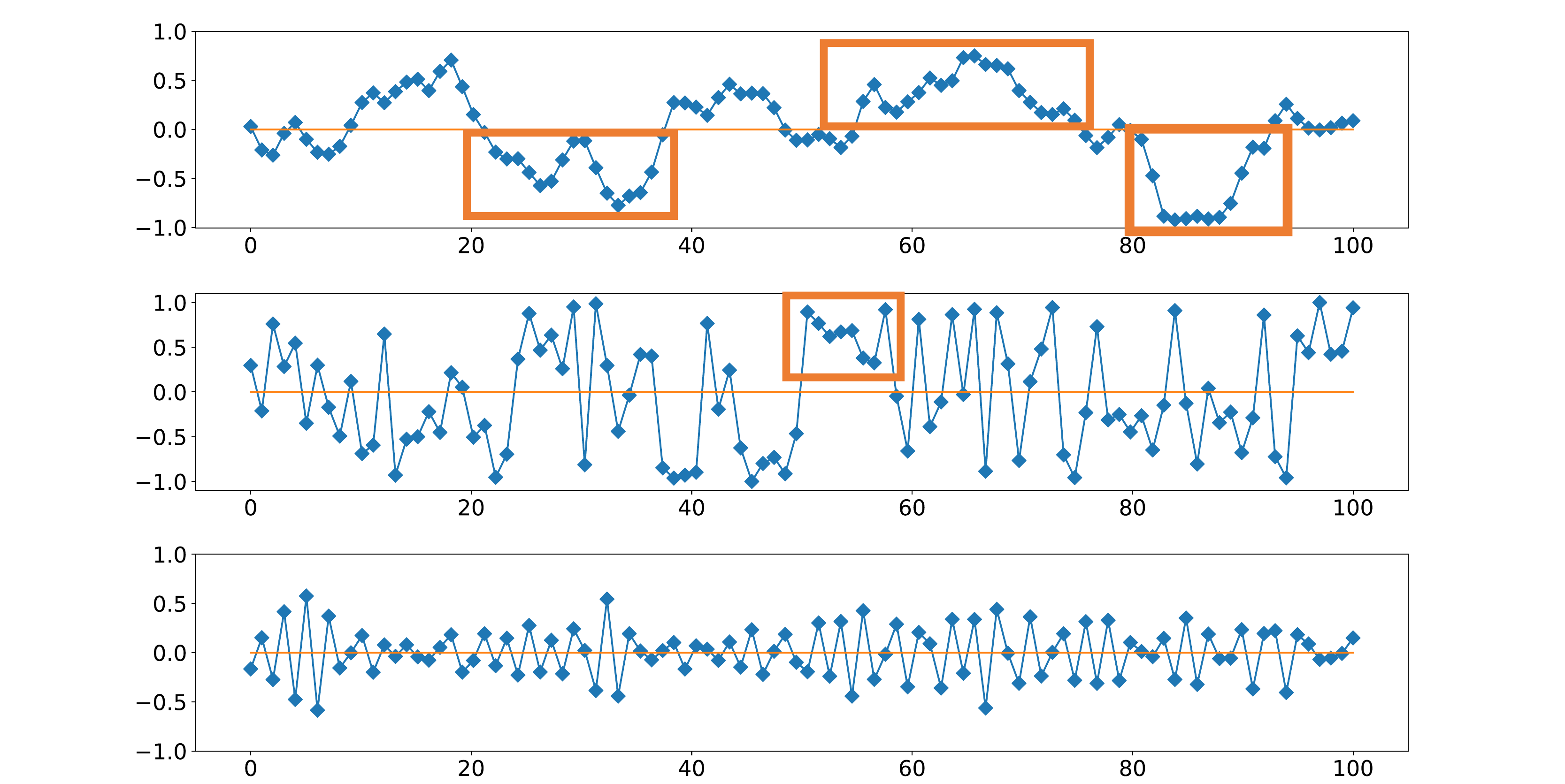}
    \caption{{\bf Correlated disorder.} From top to bottom, the three panels show one sample of three distributions with wandering exponents $w = 0.75, 0.5, 0$ respectively. The positively correlated sample features many long stretches of consecutive blocks on one side of the mean. The uncorrelated sample has fewer stretches, while the hyperuniform sample has strong local anti-correlation and hence no coherent fluctuation. These trends give a visual guide to the formulae in Sec.~\ref{sec:intro}.}
    \label{fig:CorrelatedDisorder_demo}
\end{figure*}
By varying the structure of $S(k)$ near $k = 0$, we can therefore access the full range $w \in (0,1)$ relevant for the generalized Harris bound $\nu \geq \frac{1}{d(1-w)}$. Numerically, one can sample from these correlated distributions by the following recipe: first draw a vector of independent Gaussians $\xi_x$ and then consider $l_x = W\mathcal{F}^{-1}[\sqrt{S(k)} \cdot \mathcal{F}[\vec \xi]]$ where $\mathcal{F}$ denotes a discrete Fourier transform. If we let $q_k = \mathcal{F}(l_x)$, then it is easy to check $\ev{q_k q_{-k'}} = W^2\sum_{x,y} e^{ikx} \sqrt{S(k)} e^{-ik'y} \sqrt{S(-k')} \ev{\xi_x \xi_y} = W^2 S(k)$. Therefore, the output of the algorithm $\vec l$ has the correct spatial correlation $\ev{l_x l_y} = W^2 C(|x-y|)$. Some typical samples with different wandering exponents are shown in Fig.~\ref{fig:CorrelatedDisorder_demo}. In accordance with our expectations, positive/hyperuniform correlations lead to local alignment/anti-alignment and hence a higher/lower probability for the appearance of long sequences on one side of the average. When positive correlations are too strong, coherent fluctuations of contiguous spatial clusters are heavily enhanced, leading to a smaller effective system size. As a result, we will only be able to generate reliable samples up to $w \approx 0.85$ for $\mathcal{O}(10^7)$ spatial sites.

\subsection{Motivating the MHI RG from Quantum Avalanche}\label{subsec:avalanche}

To motivate the MHI RG rules, we give a brief review of the quantum avalanche mechanism introduced in Ref.~\onlinecite{DeRoeck_Huveneers_2017}. The basic idea is to approximate each I-block as a chain of conserved l-bits and each T-block as a fully scrambled thermal bath satisfying the eigenstate thermalization hypothesis (ETH). For exponentially local interactions, the norm of operators coupling the bath to l-bits a distance $x$ away decays as $2^{-x/\zeta}$ for some decay length $\zeta$. Now take a T-block that contains $n_0$ microscopic spins. Upon coupling the T-block to nearby I-blocks, interactions have a tendency to thermalize l-bits near the boundary. But if the decay length $\zeta$ is sufficiently small, the T-block may remain trapped in a sea of l-bits. This heuristic reasoning suggests the existence of a critical $\zeta_c$ past which thermalization continues indefinitely. To derive $\zeta_c$, suppose the T-block has absorbed $n/2$ l-bits from each of the two nearby I-blocks. Then the new bath has size $n_0 + n$. In order for the l-bits further away to remain insulating, we must demand the matrix element of interactions between faraway l-bits and the original thermal bath to be much smaller than the average level spacing of the bath. By the eigenstate thermalization hypothesis (ETH), the matrix element coupling the bath and the nearest surviving l-bit is given by $\Gamma \sim \frac{2^{-n/(2\zeta)}}{\sqrt{2^{n_0 + n}}}$\footnote{Technically what appears in the denominator should be $\sqrt{e^{S(E)}}$ where $S(E)$ is the entropy density associated with a typical infinite temperature state. But for an order of magnitude estimate, it is sufficient to approximate the denominator as $2^{D_{\rm eff}}$ where $D_{\rm eff}$ is the dimension of the full Hilbert space.}. Comparing $\Gamma$ with the level spacing of the thermal bath $\delta \sim 2^{-(n_0+n)}$, we see that the T-block remains trapped iff
\begin{align}
    \frac{2^{-n/(2\zeta)}}{\sqrt{2^{n_0 + n}}} < 2^{-(n_0+n)} \quad \rightarrow \quad - \frac{n}{\zeta} + n_0 + n < 0 \,.
\end{align}
For $\zeta > 1$, the above condition can never be satisfied and the T-block absorbs more and more spins like snowballs in an avalanche. Therefore the critical value is precisely $\zeta_c = 1$. We choose our length units so that $n$ corresponds to the physical length $l^I$ of an I-block. Then by the criterion above, the shortest T-block that can thermalize an I-block of length $l^I$ is given by $d = l^I (\zeta^{-1} - 1)$. Following the convention of MHI, we refer to $d$ as the ``deficit" (see Fig.~\ref{fig:Avalanche} for a cartoon of the avalanche mechanism). Deep in the MBL phase, we expect that $\frac{d}{l^I} \rightarrow \text{const} > 0$. As the transition is approached from the MBL side, $\frac{d}{l^I} \rightarrow 0$ as T-blocks eat up larger and larger I-blocks, eventually taking over the entire spatial chain.
\begin{figure*}
    \centering
    \includegraphics[width = \textwidth]{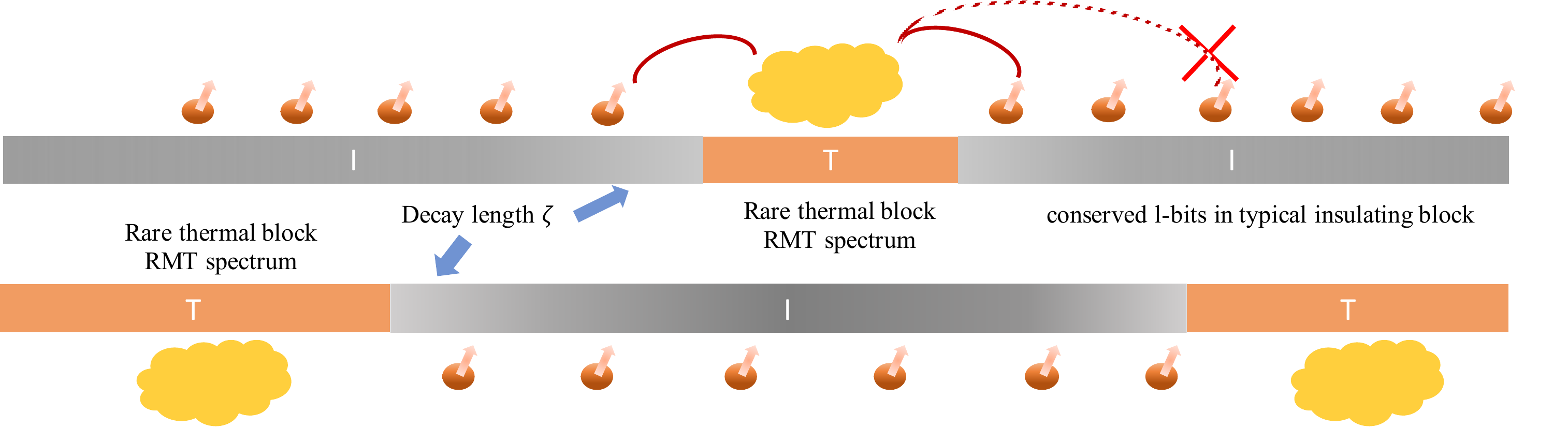}
    \caption{{\bf Avalanche picture.} The T-blocks are treated as thermal baths satisfying ETH (the yellow clouds represent a swarm of scrambled microscopic degrees of freedom). The I-blocks are a chain of exponentially localized l-bits. The red curves represent the exponentially decaying interactions between bath degrees of freedom and l-bits in I-blocks. When the decay length is too large, the bath degrees of freedom thermalize the l-bits closest to boundary, thus growing the T-block length until the entire I-block is absorbed. When the decay length is small, the T-block cannot reach far beyond the boundary and the surrounding I-blocks will remain insulating. This is the basic physical picture for the quantum avalanche of Ref.~\onlinecite{DeRoeck_Huveneers_2017}.}
    \label{fig:Avalanche}
\end{figure*}

\subsection{The Recipe for MHI RG}\label{subsec:RG_rule}

The intuitive picture in Sec.~\ref{subsec:avalanche} immediately motivates the following RG procedure:
\begin{enumerate}
    \item Consider a chain of sites labeled by an integer index $i \in \{1,\ldots, L\}$ where odd/even $i$ corresponds to T/I-blocks. For each T-block, there is a single parameter $l^T_i$ denoting the length of the T-block. For each I-block, there are two parameters $l^I_i, d_i$ denoting the physical length and the deficit length. The initial sequences $\{l^T_i\}, \{d_i\}$ are two independent correlated sequences with wandering exponent $w$ generated by the recipe in Sec.~\ref{subsec:correlated_disorder}. The initial localization length $\zeta < 1$ is chosen to be spatially uniform so that $d_i = l^I_i(\zeta^{-1} - 1)$ for all $i$. The value of $\zeta$ can be used to tune across the phase transition.
    \item At each RG step, we find either the I-block with the shortest $d_i$ or the T-block with the shortest $l^T_i$. This shortest length will be referred to as the cutoff $\Lambda$. If the shortest block is insulating, then the nearby T-blocks absorb it and acquire a total physical length $l^T_{\rm new} = l^T_{i-1} + l^I_i(=\frac{\Lambda}{x_i}) + l^T_{i+1}$ where $x_i = \zeta_i^{-1} - 1$. If the shortest block is thermal, it is too short to destabilize the nearby I-blocks and therefore gets stuck in the middle (remember that $d_i$ is the shortest T-block that can thermalize the i-th I-block and $d_i, d_{i+1} > \Lambda$). This means we get a new I-block with physical length $l^I_{\rm new} = l^I_{i-1} + l^T_{i}(=\Lambda) + l^I_{i+1}$. The new deficit involves more thought. The T-block that gets stuck in the middle is a seed for danger: an additional T-block of length $d_{i-1} -\Lambda + d_{i+1}$ could cooperate with the T-block already nested inside the new I-block to destabilize all of the l-bits in between. Thus, contrary to naive expectations, $d_{\rm new} = d_{i-1} - \Lambda + d_{i+1}$ (see Fig.~\ref{fig:RG_rule} for a pictorial representation). 
    \item After each RG step, $d_i/l^I_i$ and hence $x_i$ will not remain uniform. So we perform an additional average over all I-blocks to obtain a single value $x = \frac{1}{N_{\Lambda}}\sum_i x_i$ where $N_{\Lambda}$ is the number of blocks remaining when the cutoff is $\Lambda$. Following this, we update the deficit length to $d_i = x l^I_i$. 
    \item Steps 2 and 3 are repeated until $\Lambda$ reaches the length scale of interest. 
\end{enumerate}
At first sight, the averaging procedure in step 3 has the potential to alter critical properties. But we show that this is not the case for weak inhomogeneities in $x_i$ at some late stage in the RG: within the MBL phase, the average deficit $\ev{d}$ is much larger than $\Lambda$, and the average $\ev{l^I}$ is much larger than $\ev{l^T}$. Therefore, the RG rule for $ITI \rightarrow I$ move (which is the only move that can change $x_i$, can be approximated as $d_{\rm new} = d_{i-1} + d_{i+1}$, $l^I_{\rm new} = l^I_{i-1} + l^I_{i+1}$. This means that if $x_{i-1}, x_{i+1}$ are initially close, then $\min \{x_i,x_{i+1}\} < x_{\rm new} = (d_{i-1} + d_{i+1})/(l^I_{i-1} + l^I_{i+1}) < \max\{x_{i-1}, x_{i+1}\}$. Therefore, inhomogeneities are irrelevant under the RG flow.
\begin{figure}
    \centering
    \includegraphics[width = 0.45\textwidth]{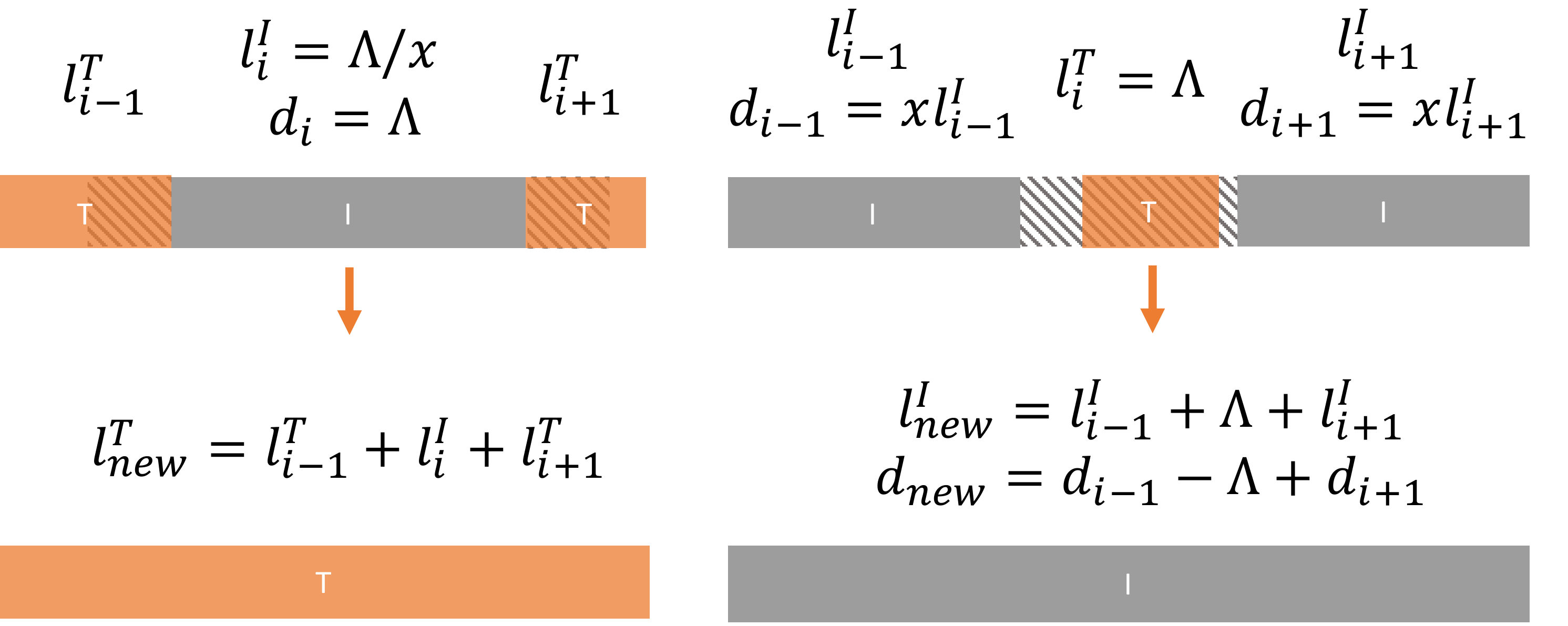}
    \caption{{\bf RG rules.} A pictorial representation of MHI RG rules. Each I-block is characterized by a deficit $d_i$ (represented by the striped region) and a physical length $l^I_i$ (striped plus gray region), while each T-block is characterized by a single physical length $l^T_i$ (the orange region). On the LHS, the T-blocks are longer than the striped region and a $TIT \rightarrow T$ follows. On the RHS, the T-block in the middle is shorter than the striped region and a $ITI \rightarrow I$ ensues.}
    \label{fig:RG_rule}
\end{figure}

\subsection{Recap of MHI Analysis and Summary of new Results}\label{subsec:overview}

Now we turn to the analysis of the RG. The complete data at each cutoff $\Lambda$ consist of a probability distribution $P_{\Lambda}(\{d_i\}, \{l^T_i\})$ for all remaining block parameters. When the initial condition is spatially uncorrelated, the RG procedure doesn't generate spatial correlations and $P_{\Lambda}$ factorizes as
\begin{equation}
    P_{\Lambda}(\{d_i\}, \{l^T_i\}) = \prod_i \rho^T_{\Lambda}(l^T_i) \mu^I_{\Lambda}(d_i) \,,
\end{equation}
where $\rho^T_{\Lambda}(l^T_i), \mu^I_{\Lambda}(d_i)$ are the single-block marginal distributions obtained from integrating over all but one of the block parameters in $P_{\Lambda}$. In the continuum limit, the RG can then be formulated as PDEs for $\rho^T_{\Lambda},\mu^I_{\Lambda}$ rather than the full $P_{\Lambda}$. This simplifying feature is essential to the RG solution in Ref.~\onlinecite{Morningstar_Huse_Imbrie_2020}. The main findings of their analysis can be summarized as follows:

(1) Eventually at infinite $\Lambda$ and within the MBL phase, $x$ converges to some value on the fixed line $\{x > 0\}$. The value of $x$ determines the fractal dimension of thermal inclusions $d_f$ and the stretching exponent of T-block lengths $\epsilon$ via $\epsilon = d_f = \frac{\log 2}{\log (x^{-1})}$. Physically, $d_f$ is defined so that a composite T-block with length $l^T$ late in the RG is made up of microscopic T-blocks with total length scaling as $(l^T)^{d_f}$. Relatedly, $\epsilon$ is defined so that $\rho^T_{\Lambda}(l) \sim e^{-\# l^{\epsilon}}$ for large $l$. Both $d_f$ and $\epsilon$ measure the difficulty of generating large T-blocks in the MBL phase and hence control dynamical properties whose leading contribution comes from rare T-block inclusions (e.g. the conductivity). 

(2) The functional RG of T-block and I-block distributions can be projected to a two-dimensional subspace spanned by the excess decay rate $x = \zeta^{-1} - \zeta_c^{-1}$ and the thermal fraction $f = \frac{\ev{l^T}}{\ev{l^T} + \ev{l^I}}$. For technical reasons, it is more convenient to replace $f$ with $y = \frac{x \Lambda^2}{\ev{d}} \rho^T_{\Lambda}(\Lambda)$, which we later demonstrate to be approximately equal to $f$ near criticality. The MBL and thermal phases meet at a critical separatrix $y(x)$ and the correlation length exponent $\nu = \infty$.

(3) The RG flow equations valid near the separatrix are given by
\begin{equation}
    \frac{dx}{d\Lambda} = - \frac{(1+x)y}{\Lambda}, \quad y_{\Lambda/x_{\Lambda}} = \left(\frac{y_{\Lambda}}{x_{\Lambda}} \right)^2 \ev{d} \mu^I_{\Lambda}(\Lambda) \,.
\end{equation}
where $\ev{d}$ the expectation value of the deficit length and $\mu^I_{\Lambda}(\Lambda)$ is the probability density of having a deficit length precisely at cutoff. In the absence of correlations, $\ev{d} \mu^I_{\Lambda}(\Lambda) \approx 1$ at large $\Lambda$ and the second equation simplifies. One can then infer the form of the separatrix $y = x^2$. For an infinitesimal perturbation $\delta_0$ away from the separatrix, the correlation length scales as $\xi \sim \delta_0^{-\log \log \delta_0^{-1}}$ which puts uncorrelated MHI in a universality class distinct from the more familiar Kosterlitz-Thouless transition (which also satisfies $\nu = \infty$). 

In the presence of correlations, technical challenges immediately arise because the full RG flow can no longer be captured by the single block marginal distributions. For hyperuniform correlations, we will give a general argument that shows their irrelevance in a class of asymptotically additive RGs (with the MHI RG as a specific example). For positive correlations, irrelevance is a special property of the MHI RG. Via a combination and analytic and numerical arguments, we will arrive at the following stability results parallel to the main findings of MHI:

(1) At infinite $\Lambda$, the fractal dimension of thermal inclusions in an I-block is $d_f \sim \frac{\log 2}{\log x^{-1}}$ while the stretching exponent for T-blocks is $\epsilon \sim \frac{\log (2-\eta)}{\log x^{-1}}$ where $0\leq \eta < 1$ is a constant that cannot be determined precisely. This shows that $\epsilon, d_f$ can behave differently, although they match in the absence of correlations. 

(2) The correlation length critical exponent $\nu = \infty$ is preserved. 

(3) With some \textit{additional technical assumptions}, the flow equations along and below the separatrix will continue to take the MHI form. However, positive correlations may potentially change the value of $\ev{d} \mu^I_{\Lambda}(\Lambda)$ but do not modify the correlation length scaling, up to nonuniversal constant coefficients that are independent of $\Lambda, \delta_0$.

\section{Stability against hyperuniform correlations}\label{sec:stability_hyp}

In this section, we define a general class of \textit{asymptotically additive RGs} (including the MHI RG as a special case) for which hyperuniform correlations in the initial block configurations do not modify the critical behavior in the vicinity of the fixed point. To arrive at this conclusion, we will argue that any initial wandering exponent $w<1/2$ always flows back to the uncorrelated value $w = 1/2$. This fragility of hyperuniform correlations is to be contrasted with the robustness of positive correlations, whose wandering exponents generally do not flow under the same class of RGs.

We first explain the basic setup. Start with $N_0$ spatial sites labeled by an index $i$, such that microscopic blocks at site $i$ are characterized by $p$ positive parameters $L_{i, \alpha = 1,\ldots, p}$ bounded from below by some initial RG cutoff $\Lambda_0$. In each step of the RG, these microscopic blocks combine to form larger composite blocks whose parameters are determined by a set of RG rules. The minimum of the updated block parameters sets the new cutoff $\Lambda$. We call the RG rules \textbf{asymptotically additive} if parameters of the new block can be written as a linear combination of parameters of the constituent blocks up to corrections subleading in $\Lambda^{-1}$. In other words, the updated block parameters approximately satisfy $L^{\rm new}_{i,\alpha} = \sum_{j, \beta} r_{\alpha\beta} L_{j,\beta}$ for some set of fixed constants $r_{\alpha\beta}$ independent of the spatial location. To give a concrete example, consider the symmetric RG of Ref.~\onlinecite{Zhang_Zhao_Devakul_Huse_2016}. If we take a perspective slightly different from Sec.~\ref{sec:intro} and view each neighboring pair of thermal and insulating blocks as living on a single spatial site, then initially we have a correlated sequence $L_{i,1} = l^T_i, L_{i,2} = l^I_i$. When the shortest block is $l^I_i$, the spatial site $i$ is eliminated and the spatial site $i+1$ has an updated thermal length $L^{\rm new}_{i+1,1} = L_{i,1} + L_{i,2} + L_{i+1,1} = l^T_i + l^I_i + l^T_{i+1}$; when the shortest block is $l^T_i$, the spatial site $i$ is eliminated and the spatial site $i-1$ has an updated insulating length $L^{\rm new}_{i-1,2} = L_{i-1,2} + L_{i,1} + L_{i,2} = l^I_{i-1} + l^T_{i} + l^I_i$. All non-vanishing components of $r_{\alpha\beta}$ are equal to one in this example and we have asymptotic additivity. One can easily verify that the GVS RG of Ref.~\onlinecite{Goremykina_Vasseur_Serbyn_2019} and the MHI RG of Ref.~\onlinecite{Morningstar_Huse_Imbrie_2020} are also asymptotically additive with $p=2, p=3$ respectively. 

We now run a general asymptotically additive RG on a spatial chain of length $N_0$. The initial block parameters are drawn from a translation-invariant probability distribution where for every $\alpha$, $\{L_{i,\alpha}\}$ is a set of $N_0$ spatially correlated parameters with mean $\ev{L_i} = L$ ($L$ is a constant $p$-component vector) and wandering exponent $0 < w < 1$. As the RG progresses to some larger cutoff $\Lambda$, the $N_0$ microscopic blocks are replaced by $N_{\Lambda}$ composite blocks with parameters $\{L_{A,\alpha}\}$ where $A = 1, \ldots, N_{\Lambda}$. It is useful to introduce a set of integers $c(A)$ monotonically increasing with $A$ such that the composite block $A$ contains all microscopic blocks with initial spatial index $i \in \{c(A),c(A)+1,\ldots, c(A+1)-1\}$. In an asymptotically additive RG, the composite block parameters are approximately equal to a linear combination of parameters of microscopic blocks with $i \in \{c(A),c(A)+1,\ldots, c(A+1)-1\}$
\begin{equation}
    L_{A,\alpha} \approx \sum_{c(A) \leq i < c(A+1)} \sum_{\beta} \tilde r_{\alpha\beta} L_{i,\beta} = \sum_{c(A) \leq i < c(A+1)} \tilde L_{i,\alpha} \,,
\end{equation}
where $\tilde L_{i,\alpha} = \sum_{\beta} \tilde r_{\alpha\beta} L_{i,\beta}$ and $\tilde r_{\alpha\beta}$ is a set of coefficients that depend on $r_{\alpha \beta}$ and the initial configuration $\{L_{i,\alpha}\}$ in some complicated way that will not be essential to the argument. 

Now we think about the consequence of this additive structure for the wandering exponent $w_{\Lambda}$ at scale $\Lambda$. Recall that $w_{\Lambda}$ is defined such that $\sigma[\mathcal{L}_{\alpha}(K)] \sim K^{w_{\Lambda}}$ where $\sigma[\ldots]$ is the standard deviation and $\mathcal{L}_{\alpha}(K) = \sum_{A=1}^K L_{A,\alpha}$ is a sum over $K$ consecutive composite blocks. Since the initial disorder distribution is translation-invariant, the average composite block size $\ev{c(A+1) - c(A)} = S$ is independent of $A$ and $\ev{\sum_{\beta} \tilde r_{\alpha\beta} L_{i,\beta}} = \tilde L_{\alpha}$ is independent of $i$. We thus have the following decomposition
\begin{equation}\label{eq:hyperuniform_key}
    \begin{aligned}
    &\sigma[\mathcal{L}_{\alpha}(K)]^2 = \ev{\left(\sum_{A=1}^K L_{A,\alpha} - KS \tilde L_{\alpha} \right)^2} \\
    &\hspace{0.5cm}= \ev{\left[c(K) \tilde L_{\alpha} - KS \tilde L_{\alpha} + \sum_{i=1}^{c(K)} \left(\tilde L_{i,\alpha} - \tilde L_{\alpha}\right)\right]^2} \\
    &\hspace{0.5cm}= \ev{\left[c(K) - KS\right]^2}\tilde L_{\alpha} + \sum_{i=1}^{c(K)} \ev{\left[\tilde L_{i,\alpha} - \tilde L_{\alpha}\right]^2} \\
    &\hspace{1cm} + 2 \tilde L_{\alpha} \ev{\left[c(K) - KS\right]\left[\tilde L_{i,\alpha} - \tilde L_{\alpha}\right]} \,.
    \end{aligned}
\end{equation}
In the last line, the first term captures fluctuations in the number of microscopic blocks contained in the $K$ composite blocks. The second term captures fluctuations in the block parameters holding the number of microscopic blocks fixed. The third term encodes correlations between these two types of fluctuations. We examine the cases $\tilde L_{\alpha} = 0$ and $\tilde L_{\alpha} \neq 0$ separately. 
\begin{enumerate}
    \item If $\tilde L_{\alpha} = 0$, then only the second term of \eq{eq:hyperuniform_key} survives. Since fluctuations in block parameters start out hyperuniform, there is a chance that $\{\tilde L_{i,\alpha}\}$ remains hyperuniform at all stages (we will see an example of this later). If that is the case, $\sigma[\mathcal{L}_{\alpha}(K)]^2 \sim c(K)^{2w} \sim K^{2w}$ and the wandering exponent $w_{\Lambda}$ does not flow.  
    \item If $\tilde L_{\alpha} \neq 0$ (which is the generic case), all three terms in \eq{eq:hyperuniform_key} compete. Initially, there is no hyperuniformity in the number fluctuations because there is no fluctuation at all. After $n \ll N_0$ RG moves, there are $N_0-2n$ blocks of size $1$ and $n$ composite blocks of size $2$. The number fluctuations are now directly associated with the fluctuations of spatial locations for the $n$ smallest numbers in a hyperuniform sequence of length $N_0$. We checked numerically that the distribution of these locations do not inherit any hyperuniformity and the number fluctuations have wandering exponent $w = 1/2$ early in the RG. What does this mean in terms of the correlation spectrum $S(k)$? The signature of hyperuniformity is a correlation hole $S(k) \sim |k|^{1 - 2 w}$ as $k \rightarrow 0$. Our argument above shows that a dilute set of block combinations already fill the correlation hole so that $w = 1/2$ and $S(k) \neq 0$ as $k \rightarrow 0$. To remove this constant term and restore the correlation hole requires an unphysical fine-tuning later on in the RG. Hence, even if the second term in \eq{eq:hyperuniform_key} retains hyperuniformity, the first term always wins since $K^{1/2} \gg K^w$ for $w < 1/2$. As a result, $w_{\Lambda} \rightarrow 1/2$ late in the RG.
    
    In contrast, positive correlations are signaled by a singularity of $S(k)$ near $k = 0$. Numerically we find that the singularity is inherited by the number fluctuations and all three terms in \eq{eq:hyperuniform_key} scale as $K^{2w}$ with $w > 1/2$. Hence, $w_{\Lambda}$ does not flow for positive correlations.
\end{enumerate}
The above casework implies that the only way to preserve hyperuniform correlations in $\mathcal{L}_{\alpha}(K)$ is to have $\tilde L_{\alpha} = 0$ and $\tilde L_{i,\alpha}$ hyperuniform at all RG stages. In fact, a simple generalization of the above argument shows that hyperuniformity can also be preserved if a linear combination $\delta_i = \sum_{\alpha} t_{\alpha} L_{i,\alpha}$ satisfies the same properties. For general RG rules with $p > 1$ and generic $r_{\alpha\beta}$ coefficients, no such special parameter can exist. We therefore conclude that hyperuniformity is irrelevant in a generic asymptotically additive RG. Since nonlinear RG rules are even more destructive to the wandering exponents, we expect the same conclusion to hold for nonlinear RGs. 

To get a concrete feel for the argument, let us consider a few examples. In the symmetric RG of Ref.~\onlinecite{Zhang_Zhao_Devakul_Huse_2016}, the basic block parameters are just the lengths $l^T, l^I$ of T/I-blocks. The RG rule is strictly additive and satisfies the assumptions in the claim:
\begin{equation}
    l^T_{\rm new} = l^T_i + l^I_{i} + l^T_{i+1} \quad l^I_{\rm new} = l^I_{i-1} + l^T_{i} + l^I_{i} \,.
\end{equation}
At criticality, $\ev{l^T} = \ev{l^I} = \mathcal{O}(\Lambda) \neq 0$ where $\Lambda$ is the moving cutoff. Therefore, hyperuniform correlations in $l^T_A, l^I_A$ get washed out by the number fluctuations. The only order parameter that has zero mean is $\delta_A = l^T_A - l^I_A$. But in general $\delta_A$ cannot be written as a linear combination of microscopic $\delta_i$ (this is easy to prove by contradiction). Therefore hyperuniform correlations are always irrelevant and $\nu(w<1/2) = \nu(w=1/2) \approx 2.5$ for every $w < 1/2$. This conclusion has been checked through finite-size scaling numerics in Appendix~\ref{sec:AppendixC}.

For the random transverse field Ising model (RTFIM) with microscopic Hamiltonian $H = \sum_i J_i Z_i Z_{i+1} + \sum_i h_i X_i$, the RG parameters for each block are $\beta_i = - \log J_i$ and $\zeta_i = - \log h_i$ with cutoff $\Gamma = \log \Omega_0 - \log \Omega$ where $\Omega = \max J_i, h_i$. Late in the RG, $\Gamma$ flows to infinity and $\beta_i, \zeta_i \geq \Gamma \geq 0$. The RG rules are still linear combinations of $\beta_i, \zeta_i$:
\begin{equation}
    \beta_{\rm new} = \beta_i - \zeta_{i+1} + \beta_{i+1} \quad \zeta_{\rm new} = \zeta_i - \beta_{i+1} + \zeta_{i+1} \,. 
\end{equation}
These RG rules are identical to the symmetric RG except for the minus signs. By our general arguments, the composite block parameters $\ev{\beta_A}, \ev{\zeta_A}$ will not remain hyperuniform at large $\Lambda$. However, the special structure of the RG rules force $\delta_A = \beta_A - \zeta_A = \sum_{i \in A} \sum_i \beta_i - \zeta_i$. Since $\delta_A$ has zero mean at criticality, we must conclude that number fluctuations do not contribute and the fluctuations of $\delta_A$ \textit{remain hyperuniform at all RG scales}! In fact, an exact solution shows that the hyperuniform RTFIM saturates the generalized Harris bound for all values of $0<w<1/2$ (see Ref.~\onlinecite{Crowley_Laumann_Gopalakrishnan_2019} for a complete analysis of this problem). 

Finally we come to the MHI RG. Clearly, the $TIT \rightarrow T$ and the $ITI \rightarrow I$ moves are both asymptotically additive. But within the MBL phase and along the critical separatrix, the average deficit $\ev{d_A}, \ev{l^T_A} \neq 0$ and hyperuniformity in $l^T_A, d_A$ is killed by the RG flow. Moreover, since the flow along the critical separatrix ends in the localized phase, there is an asymmetry between T and I-blocks such that $\ev{d_A}/\ev{l^T_A} \rightarrow \infty$. This means that there cannot be a zero-mean order parameter written as a finite linear combination of $d_A, l^T_A$. As a result, \textit{hyperuniform correlations are always irrelevant in the MHI RG}. For positive correlations, the wandering exponent remains different from the uncorrelated value for arbitrarily large $\Lambda$, potentially giving rise to a new universality class within the $\nu = \infty$ family. Whether or not this occurs will be explored in the next section. 

\section{Stability against positive correlations}\label{sec:stability_pos}

Positive correlations are generally relevant for asymptotically additive RG schemes. Nevertheless, for the MHI scheme we will find that they are irrelevant. The essential feature of the MHI RG that leads to this conclusion is that $x \to 0$ at the critical point. In what follows, we will argue for each of the three properties we previewed in Sec.~\ref{subsec:overview}: (1)~that the scaling of the fractal dimension $d_f$ is unmodified from MHI; (2)~that the correlation length exponent $\nu = \infty$ for positive correlations; and (3)~that (under some technical assumptions) the scaling of the correlation length is also unmodified from MHI. 

\subsection{Fractal dimension scaling survives correlations}
\label{subsec:stability_pos1}

We will use physical arguments to show that the structure of typical T/I-blocks near criticality is not affected by positive correlations. This analysis will not provide a concrete understanding of the flow equations, but will be sufficient to establish the more qualitative notions of stability captured by properties (1) and (2). 

The essential feature of the MHI RG that we will use is the asymmetric thermalizing powers of $T$ and $I$ blocks: while small T-blocks can easily thermalize I-blocks with large physical lengths, I-blocks must start out much larger than their neighbors to remain insulating. Deep in the MBL phase, the deficit lengths $d_i$ of the I-blocks are an appreciable fraction of their physical lengths $l^I_i$ (i.e. $x$ is not too small). As a result, a rare T-block that absorbs a neighboring I-block doesn't grow appreciably in size and has weak thermalizing power. In order to cause an instability, we would thus need to increase the thermal fraction $f = \frac{\ev{l^T}}{\ev{l^I} + \ev{l^T}}$ by seeding a critical mass of T-blocks. This implies the existence of a transition point $f_*$ corresponding to every $x_* > 0$. Now suppose we decrease the value of $x_*$, then each T-block has a higher thermalizing power and the threshold $f_*$ should decrease. As $x_* \rightarrow 0$, $f_*$ must also approach 0, because when $x = 0$, a single T-block automatically thermalizes the whole system and no MBL phase can exist. Hence the critical point is pinned at $(x,f) = (0,0)$ even in the presence of positive correlations. This argument is self-consistent as long as the fluctuations in $x_i = d_i/l^I_i$ are always much smaller than the mean, so that all the I-blocks late in the RG can be characterized by the average $x$. This self-averaging property turns out to be true everywhere outside the thermal phase: due to the asymmetric thermalizing capacities, the MBL phase (including the critical separatrix) must satisfy $\ev{L^I} > \ev{d} \gg \ev{l^T}$. Recapitulating an argument in Sec.~\ref{subsec:RG_rule}, the RG rule for $ITI \rightarrow I$ move (which is the only move that can change $x_i$), can be approximated as $d_{\rm new} = d_{i-1} + d_{i+1}$, $l^I_{\rm new} = l^I_{i-1} + l^I_{i+1}$. This means that if $x_{i-1}, x_{i+1}$ are initially close, then $\min \{x_{i-1},x_{i+1}\} < x_{\rm new} = (d_{i-1} + d_{i+1})/(l^I_{i-1} + l^I_{i+1}) < \max\{x_{i-1}, x_{i+1}\}$, implying the irrelevance of inhomogeneities in $\{x_i\}$. The existence of a well-defined separatrix even in the presence of positive correlations has an immediate implication: if we initialize the system sufficiently close to the separatrix, we will always end up in the regime where typical I-blocks have uniformly small $x$ and large physical lengths. 

As for the T-blocks, living in between these gigantic I-blocks is a huge challenge, and they have to fight for every opportunity to grow. Below the separatrix and within the $x, f \ll 1$ limit, the most efficient way to form large T-blocks is through successive $TIT \rightarrow T$ moves where $l^T_{i-1} = l^T_{i+1} = \Lambda$ and $l^I_i = \frac{\Lambda}{x}$ at every stage. The resulting fractal structure has a fractal dimension $d_f \approx \frac{\log 2}{\log (2+x^{-1})}$ which slowly approaches zero near the critical point. Now we would like to argue that typical T-blocks late in the RG have precisely this structure. If the T-block lengths $l^T_i$ were independently distributed, then the probability of growing a fractal T-block of length $l$ scales as $\exp{-l^{d_f}} \sim \exp{-l^{\frac{\log 2}{\log x^{-1}}}}$. This is to be contrasted with the probability of having a non-fractal T-block of length $l$ which scales as $\exp{-l}$. As $x \rightarrow 0$, $\exp{-l^{d_f}} \gg \exp{-l}$ and hence fractal regions dominate late in the RG, precisely as shown in the uncorrelated MHI analysis~\cite{Morningstar_Huse_Imbrie_2020}. In the correlated case, to establish a similar dominance, we need two crucial ingredients: (a) The probability of having a pair of neighboring T-block at cutoff $C^{TT}_{\Lambda}(\Lambda,\Lambda)$ should approximately factorize into $\rho^T_{\Lambda}(\Lambda)^2$ where $\rho^T_{\Lambda}(l)$ is the marginal distribution of single T-block lengths. (b) The presence of a T-block at cutoff should not be strongly correlated with the presence of a neighboring I-block at cutoff. This avoids the appearance of a long chain $TITI\ldots T$ where all I-blocks are at cutoff and the whole chain merges into a single T-block with $\mathcal{O}(1)$ fractal dimension. 

To argue for these ``factorization"-type results, we again take advantage of the asymmetric thermalizing powers of T and I-blocks. Let us consider two composite T-blocks with length $L^T_A, L^T_{A+1}$, each containing $\mathcal{O}(\Lambda)$ microscopic blocks. Then as $x \rightarrow 0$, the composite I-block sandwiched by the T-blocks contains at least $\mathcal{O}(\frac{\Lambda}{x})$ microscopic blocks, reflecting the asymmetry. If we denote the microscopic block lengths by $l_i$ and consider a correlation function $C(i,j) \sim 1/|i-j|^c$, then by the general arguments of Sec.~\ref{sec:stability_hyp}, block length correlations compete with number fluctuations and the covariance of $L^T_A, L^T_{A+1}$ scales as
\begin{equation}
    \begin{aligned}
    \ev{L^T_A L^T_{A+1}}_{\rm conn} &\approx \ev{\left(\sum_{i=1}^{\Lambda} l_i\right)\left(\sum_{j=1}^{\Lambda} l_{\mathcal{O}(\frac{\Lambda}{x}) + j}\right)}_{\rm conn} \\
    &\sim \sum_{i,j =1}^{\Lambda} \frac{1}{\left|\mathcal{O}(\frac{\Lambda}{x}) + j - i\right|^c} \lesssim \Lambda^2 (\frac{\Lambda}{x})^{-c} \,.
    \end{aligned}
\end{equation}
On the other hand, the variance of an individual composite block $L^T_A$ is
\begin{equation}
    \ev{L^T_A L^T_A}_{\rm conn} = \ev{\left(\sum_{i=1}^{\Lambda} l_i\right)\left(\sum_{j=1}^{\Lambda} l_{\Lambda+j}\right)}_{\rm conn} \sim \Lambda^{2-c} \,.
\end{equation}
Comparing the two estimates above, we see
\begin{equation}
    \ev{L^T_A L^T_{A+1}}_{\rm conn} \sim x^c \ev{L^T_A L^T_A}_{\rm conn} \,,
\end{equation}
implying that for every $0 < c < 1$, the correlations between nearby T-blocks are asymptotically suppressed in the limit $x \rightarrow 0$. This argument easily generalizes to multi-point correlations between distant composite T-blocks, giving the estimate $\frac{\ev{(L^T_A)^n (L^T_{A+B})^n}_{\rm conn}}{\ev{(L^T_A)^{2n}}_{\rm conn}} \sim x^c/B^c$. Hence, we have a robust conclusion that the wandering exponent of T-blocks $w_T \rightarrow 1/2$ as $x \rightarrow 0$, and the joint distributions of multiple consecutive T-blocks should factorize into products of marginals, giving an even stronger version of ingredient (a). In contrast, the fluctuations of I-block lengths retain the wandering exponent $w_I > 1/2$ of UV correlations. This is because the T-block in between nearby I-blocks is negligibly short and the asymptotic RG move is just successive I-block additions $L^I_{\rm new} = L^I_A + L^I_{A+1}$ which preserve the wandering exponent, as we have shown in Sec.~\ref{sec:stability_hyp}. 

For ingredient (b), consider now nearby T and I-blocks containing $\Lambda$ and $\Lambda/x$ microscopic blocks respectively. Imitating the calculation before, we have
\begin{equation}
    \ev{(L^T_A)^n (L^I_{A+1})^n} - \ev{(L^T_A)^n} \ev{(L^I_{A+1})^n} \lesssim [\Lambda (\frac{\Lambda}{x})^{1-c}]^n\,.
\end{equation}
Rewriting these correlators in terms of joint and marginal distributions and dividing by a uniform factor $x = \frac{d}{L^I}$, we have
\begin{equation}\label{eq:moment_bound}
    \int l^n d^n C^{TI}_{\Lambda,c}(l,d) \sim \left(\frac{x}{\Lambda}\right)^{cn} [\int l^n \rho^T_{\Lambda}(l)]\cdot [\int d^n \mu^I_{\Lambda}(d)] \,.
\end{equation}
The above moment estimates show that nearby T and I-blocks become weakly correlated late in the RG, thereby establishing a quantitative formulation of ingredient (b). Moreover, they motivate a stronger pointwise bound $C^{TI}_{\Lambda,c}(l,d) \ll \rho^T_{\Lambda}(l) \mu^I_{\Lambda}(d)$ although no rigorous proof can be given in the absence of additional regularity assumptions. For concreteness, we provide in Fig.~\ref{fig:Test_CTI} some numerical evidence for this pointwise estimate evaluated at the cutoff $d=\Lambda$. Due to the nature of $\nu = \infty$ RGs, the critical window is extremely narrow and we cannot truly approach the $x \ll 1$ regime even for very large system sizes ($\sim 4 \cdot 10^6$). But the trend of decaying correlations in our numerics is consistent with all the analytic arguments. We will use these facts again in the analysis of Sec.~\ref{subsec:stability_pos3}.
\begin{figure}
    \centering
    \includegraphics[width = \textwidth/2-9pt]{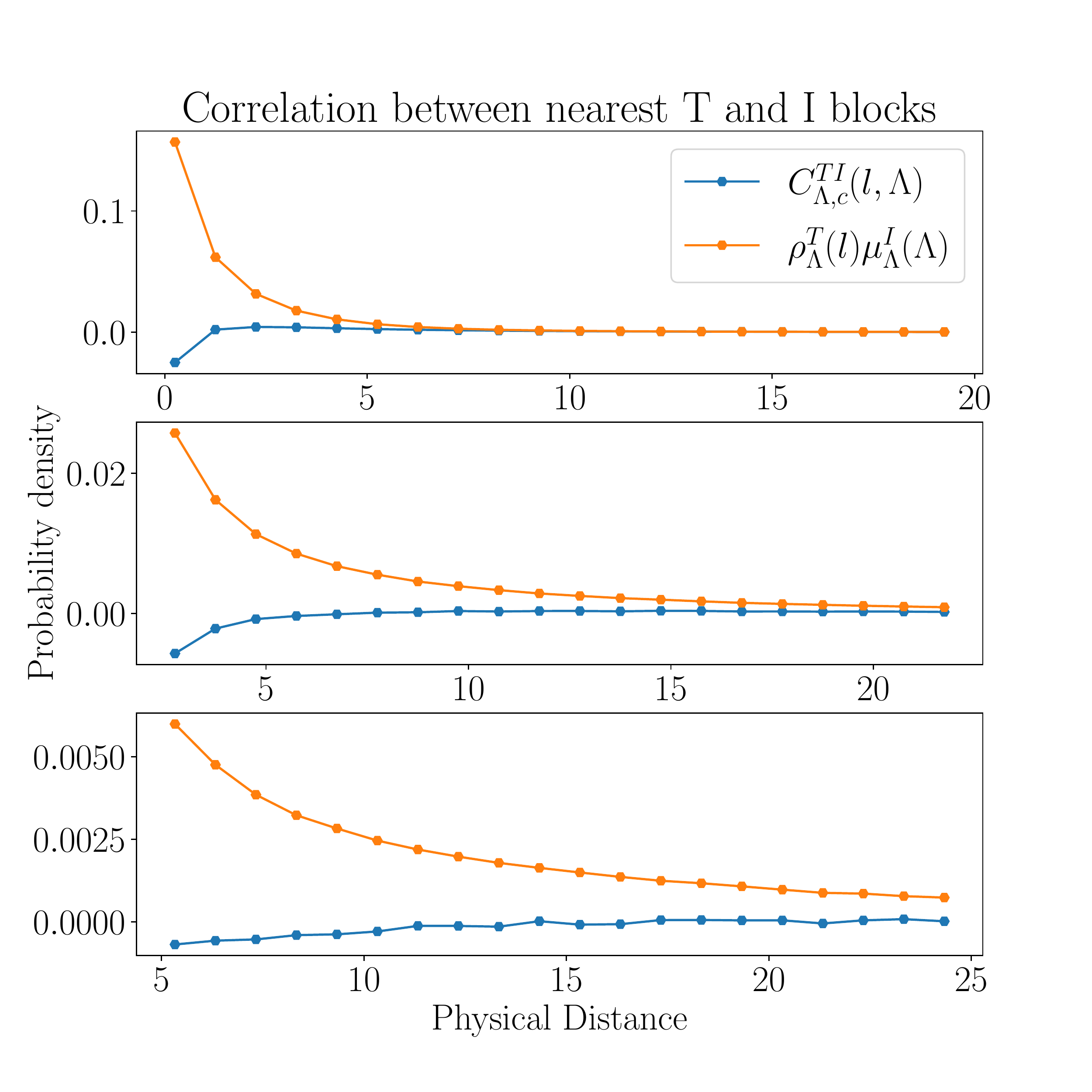}
    \includegraphics[width = \textwidth/2-9pt]{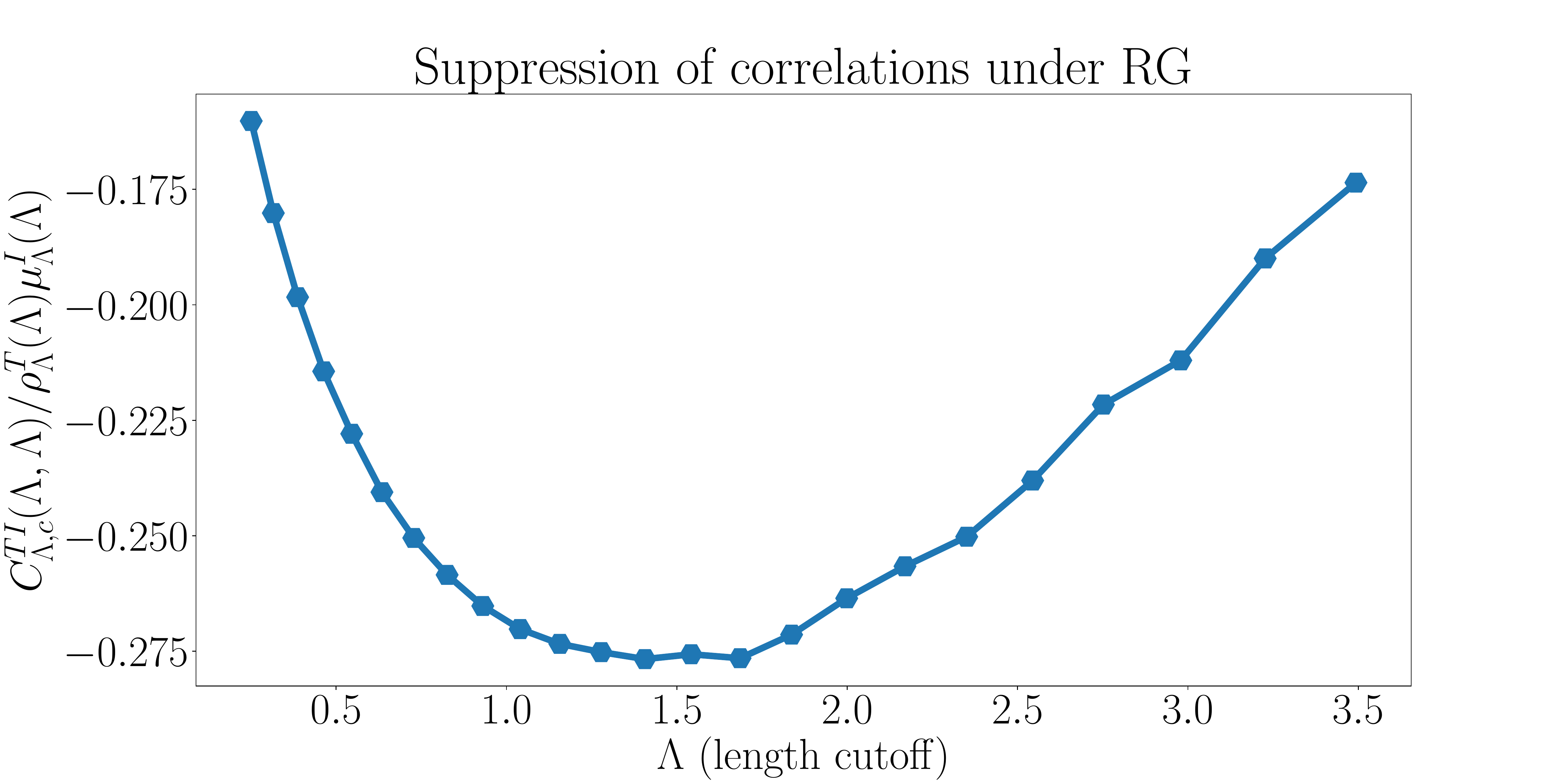}
    \caption{{\bf Suppression of the TI correlations.} In the top figure we provide some numerical evidence for \eq{eq:moment_bound} by plotting $C^{TI}_{\Lambda,c}(l,\Lambda)/\rho^T_{\Lambda}(l) \mu^I_{\Lambda}(\Lambda)$ against the physical separation $l$. We initiated the RG with $4 \cdot 10^6$ blocks and the top, middle, bottom panels are snapshots taken when $2 \cdot 10^6, 8 \cdot 10^5, 1.4 \cdot 10^5$ blocks remain. In the bottom figure, we provide an alternative visualization by fixing $l=d=\Lambda$ and tracking the evolution of $C^{TI}_{\Lambda,c}(\Lambda,\Lambda)$ as a function of $\Lambda$. One can clearly see that the joint distribution remains close to the product of marginals as the RG progresses, consistent with \eq{eq:moment_bound}. The results for $C^{TT}_{\Lambda,c}$ and $C^{II}_{\Lambda,c}$ are qualitatively similar.}
    \label{fig:Test_CTI}
\end{figure}

With ingredients (a) and (b) in hand, we return to the argument about fractal dimensions. In the presence of UV positive correlations with decay exponent $c$, the probability of growing a non-fractal T-block with large length $l$ scales as $\exp{-l^c}$ where $c$ is the decay exponent of the correlations. By the factorization argument above, the rare T-blocks in the IR are asymptotically independent and the probability of growing a fractal inclusion of length $l$ retains its uncorrelated scaling $\exp{-l^{\frac{\log 2}{\log x^{-1}}}}$. Comparing $\exp{-l^{\frac{\log 2}{\log x^{-1}}}}$ with $\exp{-l^c}$, we see that for all $0<c<1$, taking $x \ll 1$ always guarantees that the fractal inclusions eventually dominate over the non-fractal rare regions, thereby establishing property (1). At first sight, one may guess that $d_f$ is equal to the stretching exponent $\epsilon$ for the T-block distribution because the number of independent rare events needed to form a rare T-block with length $l$ scales as $l^{d_f}$. But the rarity of those events changes with the RG scale (as pointed out by Ref.~\onlinecite{Morningstar_Huse_2019}) and with the flow of correlations. Hence, no precise relationship between $\epsilon$ and $d_f$ can be inferred, although the singular scaling with $x$ should be the same. With additional technical assumptions, we will verify in Sec.~\ref{subsec:stability_pos3} that this is indeed the case. In summary, the most dangerous thermalizers that prevent the rapid decay of $f$ near the critical separatrix consist of typical T-blocks in the $x,f \ll 1$ regime that are fractals with small fractal dimension. Positive correlations might affect the statistics of rare dense T-blocks with larger fractal dimensions, but they play no role when the system is sufficiently close to the critical point. 

\subsection{Correlation-length exponent \texorpdfstring{$\nu = \infty$}{} survives correlations}\label{subsec:stability_pos2}

The structure of typical T and I-blocks described above also helps us understand the correlation length exponent $\nu$. Let us define the RG time $t = t_0 + \log \Lambda$ which keeps track of the exponentially large physical time elapsed during the block combination RG moves. At $t = 0$, we initialize the RG very close to the critical separatrix and with $x, f \ll 1$. In the two dimensional space $(x,f)$, the separatrix near the fixed point $(0,0)$ can be approximated by a curve $f(x) = x^{\beta}$ where $\beta$ is an undetermined coefficient. To define the correlation length scaling, we slightly perturb away from the separatrix so that $f(x_0) = x_0^{\beta+\delta_0}$ for some small $\delta_0 > 0$ and ask at what RG length scale $\Lambda$ does $\delta(t)$ grow to an $\mathcal{O}(1)$ number.  Within this framework, $\nu = \infty$ indicates the failure of a standard scaling ansatz $\Lambda(\delta_0) = \delta_0^{-\nu}$. For example, the Kosterlitz-Thouless transition obeys $\Lambda(\delta_0) \sim e^{b\delta_0^{-1/2}}$ and the uncorrelated MHI transition obeys $\Lambda(\delta_0) \sim \delta_0^{-\log \log \delta_0^{-1}}$. In both cases, $\Lambda(\delta_0)$ grows faster than any power law in $\delta_0$, invaliding the hypothesis of finite $\nu$. 

Now suppose we start on the separatrix of the correlated MHI RG and slightly increase the value of $x$ to stay in the $x \ll 1$ limit while moving below the critical separatrix. The new starting point can be regarded as a small perturbation to the correlated RG $f_0 = x_0^{\beta+\delta_0}$ where $\delta_0$ depends smoothly on the shift of $x$. Note that this shift doesn't change the structure of fluctuations in the T/I-blocks and the separatrix itself doesn't shift. But since positive correlations enhance coherent fluctuations, it should be more difficult for the RG flow to bring the system off criticality and the correlation length should diverge faster with $\delta_0$. Writing the correlation length exponent $\nu(w)$ as a function of the wandering exponent $w$, we then expect $\nu(w > 1/2) \geq \nu(w = 1/2)$. In the $x, f\ll 1$ regime, the uncorrelated MHI RG already has $\nu = \infty$. Hence we expect $\nu = \infty$ for positive correlations as well. 

To illustrate this general picture, we can analyze the fractal inclusions that drive the critical fluctuations more carefully. Each time a rare I-block at cutoff gets absorbed by the nearest T-blocks, the new T-block that forms has length $\mathcal{O}(\frac{\Lambda}{x})$. Slightly off criticality and in the MBL side of the separatrix, these are the dominant processes that prevent T-blocks from completely vanishing. Hence we can basically run the RG in discrete steps, where the cutoff gets moved from $\Lambda \rightarrow \frac{\Lambda}{x}$ in each step. If we denote the number of such discrete RG steps by $n$, then $\frac{\Delta t}{\Delta n} = \log x^{-1}$. Along the separatrix, due to critical slowing down, $x(t)$ is an inverse power law in $t$ and $\log x^{-1} \sim \log t$. Therefore, upon integration, the total RG time $T$ is related to the total number of discrete RG steps by $T \sim N \log N$ up to subleading corrections. Now we start with a small deviation $\delta_0$ from the separatrix and suppose that $\delta(t)$ becomes $\mathcal{O}(1)$ when $t = T(\delta_0)$. Using the definition of $\nu$, we then conclude that
\begin{equation}
    \Lambda(\delta_0) \sim e^{T(\delta_0)} \sim e^{N(\delta_0) \log N(\delta_0)} \sim \delta_0^{-\nu} = e^{\nu \log \delta_0^{-1}}.
\end{equation}
If $\nu < \infty$, the above equation implies $\delta(n) \sim e^{n \log n}$, which is faster than the exponential growth $\delta(n) \sim e^n$ seen in the uncorrelated MHI RG. In order to have such a super-exponential growth, the fractal thermal inclusions controlling the transition would have to be easier to suppress in the system with positive correlation than in an uncorrelated system. This is opposite from the physical intuition that positive correlations enhance coherent fluctuations (i.e. in this case the coherent fluctuations are just the random production of larger and larger fractal T-blocks). As a result, $\nu = \infty$ continues to hold for positive correlations and property (2) is established. One caveat of the above reasoning is that positive correlations modify the location of the separatrix and a direct comparison of the scaling ansatzes for uncorrelated and correlated RGs is not strictly justified because $\delta_0$'s are defined with different reference points. We will address these and other subtle issues in the next section.

\subsection{Precise argument for correlation-length scaling based on the hierarchy of flow equations}\label{subsec:stability_pos3}

The preceding intuitive discussion explains why positive correlations are irrelevant on the level of typical fractal block structures and correlation length exponent $\nu$. However, even if these results are true, there is no reason to expect that the universality class of the transition also remains unmodified. As we have seen, the correlation length scaling is determined by the rate at which a small perturbation $\delta_0$ away from the critical separatrix increases with the number of fractal steps $n$. Since positive correlations are expected to suppress the growth of $\delta_0$ with $n$, it is in principle possible that they modify the uncorrelated MHI scaling $\delta(n) \sim e^n$ to the KT scaling $\delta(n) \sim \log n$ observed in earlier RG studies\cite{Goremykina_Vasseur_Serbyn_2019, Dumitrescu_Goremykina_Parameswaran_Serbyn_Vasseur_2019}. In fact, as we will see, when connected correlators $C^{TI}_{\Lambda,c}(l, d)$ and their higher-order analogues are sufficiently large, the asymptotic flow equations projected to the two-parameter space will indeed have a different structure. Nevertheless, under suitable assumptions that are supported by analytics and numerics, the change in flow equations leaves the exponential scaling of $\delta(n)$ invariant, thereby confirming property (3). 

We begin our analysis by introducing some notations. As previewed in Sec.~\ref{subsec:overview}, for a general correlated initial distribution, the joint probability distribution $P_{\Lambda}(\vec l^T, \vec d)$ does not factorize into single-block marginal distributions. However, since the RG rules are local in space, $C^{(n)}_{\Lambda}(l_1,d_1,\ldots)$ (the probability density that a contiguous chain of n blocks have lengths $(l_1,d_1,\ldots)$ when the cutoff is at $\Lambda$) would only be sensitive to correlations between the chain and the closest two blocks to the left and right of the chain. As shown in Appendix~\ref{sec:AppendixA}, formalizing this intuition leads to an infinite hierarchy of equations that structurally resemble the BBGKY hierarchy in classical statistical mechanics 
\begin{equation}
    \partial_{\Lambda} C^{(n)}_{\Lambda}(l_1,d_1,\ldots) = F_{\text{depl}}[C_{\Lambda}^{(n+2)}] + F_{\text{prod}}[C_{\Lambda}^{(n+2)}] \,.
\end{equation}
Physically, the depletion term $F_{\text{depl}}$ accounts for RG moves where a chain of n blocks with length $(l_1,d_1,\ldots)$ at cutoff $\Lambda$ is no longer present at cutoff $\Lambda + d \Lambda$ due to decimations that modify one of the block lengths. In contrast the production term $F_{\text{prod}}$ encodes RG moves that create a new chain of n blocks with lengths $(l_1,d_1,\ldots)$ not present at cutoff $\Lambda$.

While this infinite hierarchy of equations is difficult to solve, a lot of progress can be made by concentrating on the few-body correlations. Following the convention of Ref.~\onlinecite{Morningstar_Huse_Imbrie_2020}, we first consider the marginal distributions $\rho^T_{\Lambda}(l), \mu^I_{\Lambda}(d)$ of single-block lengths obtained from integrating out all but one of the blocks in the full joint probability distribution $P(\vec l^T, \vec d)$ (by translation invariance the choice of blocks to integrate over doesn't matter). From Appendix~\ref{sec:AppendixA} we quote the following flow equations
\begin{align}
    \partial_{\Lambda} \rho^T_{\Lambda}(l) &= \rho^T_{\Lambda}(l)\left[\mu^I_{\Lambda}(\Lambda) + \rho^T_{\Lambda}(\Lambda)\right] - 2C^{TI}_{\Lambda}(l, \Lambda) \notag \\
    &+ \int_{\Lambda}^{l - \frac{\Lambda}{x} - \Lambda} C^{TIT}_{\Lambda}(l_1, \Lambda, l - \frac{\Lambda}{x} - l_1) d l_1 \label{eq:rhoT_flow} \,,\\
    \partial_{\Lambda} \mu^I_{\Lambda}(d) &= \mu^I_{\Lambda}(d)\left[\mu^I_{\Lambda}(\Lambda) + \notag \rho^T_{\Lambda}(\Lambda)\right] - 2C^{TI}_{\Lambda}(\Lambda, d) \\
    &+ \int_{\Lambda}^{d} C^{ITI}_{\Lambda}(s, \Lambda, d + \Lambda - s) d s \label{muI_flow} \,.
\end{align}
These equations differ from MHI due to the appearance of two-block and three-block joint distributions on the RHS. In the absence of correlations, all joint distributions factorize into products of marginals and we recover a closed set of PDEs for $\rho^T_{\Lambda}(l), \mu^I_{\Lambda}(d)$ as in MHI. Once we turn on correlations, the equations no longer close, and we need to do more work. Fortunately, from Sec.~\ref{subsec:stability_pos1}, we know that the phase transition is controlled by a competition between the tendency of T-blocks to proliferate and the presence of an excess decay rate $x > 0$ that protects l-bits in the I-block. This competition can be understood by projecting the infinite-dimensional flow to a two-dimensional subspace $(x,f)$ as long as there is no additional relevant RG direction. We verify this by plotting the numerical RG flow lines for different initial conditions and showing that they have no crossing down to the largest length scales. An example of this numerical check is shown in Fig.~\ref{fig:MH_poscorr_flowlines} for $w = 0.75$ and initial system size $L = 4 \cdot 10^6$.
\begin{figure}
    \centering
    \includegraphics[width = 0.48\textwidth]{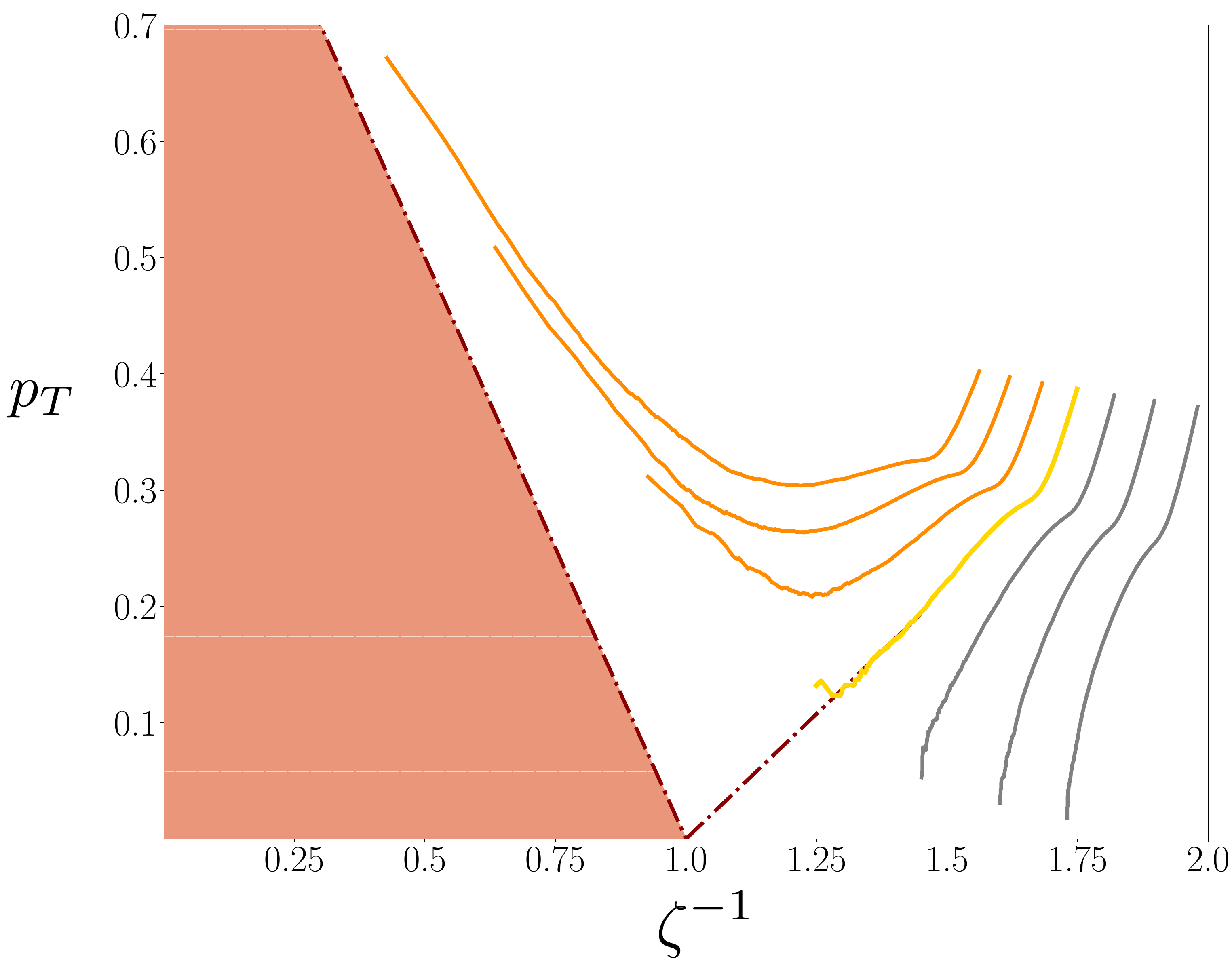}
    \caption{{\bf Numerical RG flow.} Numerical RG flow lines of MHI RG with $w = 0.75$ and initial system size $L = 4 \cdot 10^6$. The different flows correspond to initial values of $\zeta$ between $0.245 \sim 0.305$ in steps of $0.01$. The yellow critical flow separates the orange/grey flows which land in the T/I phase respectively. The ultra-thermal shaded region is inaccessible given the finite density of l-bits to start with.}
    \label{fig:MH_poscorr_flowlines}
\end{figure}

Since the thermal fraction $f$ is awkward to work with for technical reasons, we introduce a rescaled variable $y$ defined as follows:
\begin{equation}
    x = \frac{\ev{d}}{\ev{l^I}} \quad y = \Lambda^2 r_{\Lambda}(\Lambda) \quad r_{\Lambda}(l) = \frac{x}{\ev{d}} \rho^T_{\Lambda}(l) \,.
\end{equation}
In the insulating phase, 
\begin{equation}
    f = \frac{\ev{l^T}}{\ev{l^I} + \ev{l^T}} \approx \frac{\ev{l^T}}{\ev{l^I}} = \frac{x}{\ev{d}} \int_{\Lambda}^{\infty} l \rho^T_{\Lambda}(l) d l \,. 
\end{equation}
Along the separatrix, we will later show that $\Lambda^2 \rho^T_{\Lambda}(\Lambda) \approx \ev{l^T}$, implying the asymptotic equivalence between $f$ and $y$ at large $\Lambda$. Below the separatrix, $\ev{d}$ grows as a stretched exponential in $\Lambda$ and $y, f$ both tend to zero. Therefore, one can loosely think about $y$ as a \textbf{proxy for the thermal fraction}. These and other notations are summarized in Table~\ref{tab:notation} for convenient reference.

Some qualitative features of this projected RG flow are now transparent. In the insulating phase, $f \rightarrow 0$ and $y \rightarrow 0$ as I-blocks dominate over T-blocks. A finite excess decay rate $x$ persists to infinite $\Lambda$ and we land somewhere on the MBL fixed line $x > 0, y = 0$; in the thermal phase, $f \rightarrow 1$ and $y \rightarrow \infty$ as $\ev{l^T} \gg \ev{L^I}$. Therefore, we expect a \textbf{critical separatrix} in $(x,y)$ marking a phase transition. To study this separatrix and the critical scaling close to it, we need to derive flow equations for $x$ and $y$. The flow equation for $x$ requires only a single lemma:
\begin{lemma}\label{lemma:1pt_function_flow}
    The form of the flow equations for marginal expectation values $\ev{l^T}, \ev{L^I}, \ev{d}$ are unmodified by correlations.
\end{lemma}
The proof of this lemma involves a tedious calculation which we include in Appendix~\ref{sec:AppendixA}. Here, we will merely quote the results:
\begin{align}
    \frac{d \ev{L^I}}{d\Lambda} &= \mu^I_{\Lambda}(\Lambda)\left[\ev{L^I} - \frac{\Lambda}{x}\right] + \rho^T_{\Lambda}(\Lambda) \left[\ev{L^I} + \Lambda\right]\label{eq:LI_flow} \,,\\
    \frac{d \ev{l^T}}{d\Lambda} &= \rho^T_{\Lambda}(\Lambda)\left[\ev{l^T} - \Lambda\right] + \mu^I_{\Lambda}(\Lambda) \left[\ev{l^T} + \frac{\Lambda}{x}\right] \label{eq:LT_flow}\,,\\
    \frac{d \ev{d}}{d\Lambda} &= \left[\mu^I_{\Lambda}(\Lambda) + \rho^T_{\Lambda}(\Lambda)\right](\ev{d} - \Lambda) \label{eq:d_flow} \,.
\end{align}
It is important to remark that this is a \textbf{not} a closed set of equations for the averages $\ev{L^I}, \ev{l^T}, \ev{d}$ because the RHS involves the marginal distributions. Therefore, although these equations are formally equivalent to those in MHI, they are sensitive to correlations through the flow of single-block marginals on the RHS. 

\begin{widetext}
\renewcommand\arraystretch{2}
\begin{longtable}{|P{1.5cm}|P{2.4cm}|P{3cm}|P{2.8cm}|P{2.2cm}|P{2.2cm}|P{2cm}|}
\caption{Notation for important variables}
\label{tab:notation}\\
\hline
\textbf{T-block length} & \textbf{I-block length, deficit length} & \textbf{Excess interaction decay rate} & \textbf{Proxy for thermal fraction} & \textbf{Marginal distribution} & \textbf{Joint distribution} & \textbf{Fractal dimension} \\
\hline
$l^T$ or $l$ & $l^I, d$ & $x = \frac{\ev{d}}{\ev{l^I}} = \zeta^{-1} - 1$ & $y = \Lambda^2 r_{\Lambda}(\Lambda)$ & $\rho^T_{\Lambda}(l), \mu^I_{\Lambda}(d)$ & $C^{TI\ldots T}_{\Lambda}(l_1,\ldots)$ & $d_f$ \\
\hline
\end{longtable}
\end{widetext}

With this difference in mind, we proceed to work out the flow equations for $x$ and $y$. The flow equation for $x$ follows easily from the flow equations for expectation values derived above:
\begin{equation}\label{eq:x_flow}
    \begin{aligned}
    \frac{dx}{d \Lambda} &= \frac{1}{\ev{l^I}} \frac{d \ev{d}}{d\Lambda} - \frac{\ev{d}}{\ev{l^I}^2} \frac{d \ev{l^I}}{d\Lambda} \\
    &= \frac{1}{\ev{l^I}} \left[\mu^I_{\Lambda}(\Lambda) + \rho^T_{\Lambda}(\Lambda)\right] \left(\ev{d} - \Lambda\right) \\
    &- \frac{\ev{d}}{\ev{l^I}^2} \big(\left[\ev{l^I} - \Lambda/x\right] \mu^I_{\Lambda}(\Lambda) + \left(\ev{l^I} + \Lambda\right) \rho^T_{\Lambda}(\Lambda)\big) \\
    &= \frac{-\Lambda (1+x) x \rho^T_{\Lambda}(\Lambda)}{\ev{d}} = - \frac{(1+x) y}{\Lambda} \,.
    \end{aligned}
\end{equation}
The flow of $y = \Lambda^2 r_{\Lambda}(\Lambda)$ will follow from the flow of $r_{\Lambda}(l)$, which is simple to derive using the flow of $\rho^T_{\Lambda}(l), \ev{d}$ in \eq{eq:rhoT_flow}, \eq{eq:d_flow}(see Appendix~\ref{sec:AppendixB} for details):
\begin{equation}\label{eq:rLambda_l_flow}
    \begin{aligned}
        \partial_{\Lambda} r_{\Lambda}(l) &= \bigg(- \frac{y}{\Lambda} + \frac{\Lambda \mu^I_{\Lambda}(\Lambda)}{\ev{d}} - \frac{2 C^{TI}_{\Lambda}(l,\Lambda)}{\rho^T_{\Lambda}(l)} \bigg) r_{\Lambda}(l) \\
        &+ \frac{x}{\ev{d}} \int_{\Lambda}^{l - \Lambda(1+x^{-1})} C^{TIT}_{\Lambda}\big(l_1,\Lambda, l-\Lambda/x - l_1\big) d l_1  \,.
    \end{aligned}
\end{equation}
In the absence of correlations, MHI was able to integrate this flow and obtain a recursion relation that estimates $r_{\Lambda}(\Lambda/x)$ based on knowledge of $r_{\Lambda}(\Lambda)$. This recursion, combined with the flow of $x$, then gives a complete understanding of the critical separatrix and small perturbations around it. Crucially, this recursion relies again on the fact that $C^{TIT}_{\Lambda}$ factorizes into a product of marginals so that the LHS and RHS can be related to the marginals evaluated at different RG scales. In the presence of correlations, factorization is no longer possible and a recursion of $r_{\Lambda}$ requires a different argument which we now summarize.

First we make the general decomposition $C_{\Lambda} = C_{\Lambda, \text{disc}} + C_{\Lambda, c}$ where $C_{\Lambda, \text{disc}}$ is a product of marginal distributions and $C_{\Lambda,c}$ is the connected part which vanishes in the absence of correlations. Under \textbf{three fundamental assumptions}, we will show that the first term on the RHS of \eq{eq:rLambda_l_flow} is negligible when we integrate 
from $\Lambda = x^2 l$ and $\Lambda = xl$, \emph{even when the connected correlators $C_{\Lambda,c}$ are larger than $C_{disc}$}. Within the same integration range and using the same assumptions on $C_{\Lambda,c}$, we then argue that for positive correlations, although the integrand cannot be factorized, the full integral nonetheless reduces to the factorized answer (an intuitive justification of this fact will be provided along the way). These arguments would produce a recursion in $r_{\Lambda}$ that differs from the MHI answer. In the final step we show that the modified recursion can lead to a shifted stretching exponent $\epsilon$ inside the MBL phase but cannot change the fractal dimension $d_f$ or the correlation length scaling, providing a precise extension of the results in Sec.~\ref{subsec:stability_pos1}. Throughout the argument, we will state various technical lemmas and provide some intuition. But the proofs for these lemmas are relegated to Appendix~\ref{sec:AppendixB}. 

The first assumption we make is a generic property of systems driven to criticality:

\textit{Assumption 1: \textbf{Critical slowing down} holds along the separatrix so that $x \sim t^{-\alpha}$ for some positive exponent $\alpha$ where $t = \log \Lambda$ is the RG time. In the uncorrelated MHI RG, $\alpha = 1$. Here we only assume that $\alpha$ is finite. }

This is physically very reasonable because fluctuations become more and more macroscopic near the critical point at $(x,y) = (0,0)$ and relaxation of these macroscopic regions to equilibrium takes longer and longer. In the language of renormalization group, the fastest term in the $\beta$-function (which gives rise to exponentially fast growth/decay) vanishes along the critical separatrix and the subleading terms take over to give a power law behavior. Surprisingly, critical slowing down alone already provides a powerful constraint:
\begin{lemma}\label{lemma:1pt_function_crit}
    Under assumption 1, $\ev{d} \sim \Lambda \log \Lambda \gg \Lambda$, $\mu^I_{\Lambda}(\Lambda) < x\rho^T_{\Lambda}(\Lambda)$, $\rho^T_{\Lambda}(\Lambda) \approx \frac{1}{\Lambda} + \mathcal{O}(\frac{1}{\Lambda \rm{Poly}(\log \Lambda)})$ in the large $\Lambda$ limit along the critical separatrix. Below the critical separatrix (in the MBL phase), $\ev{d} \gg \Lambda$ and $\mu^I_{\Lambda}(\Lambda) < x \rho^T_{\Lambda}(\Lambda)$ still hold. 
\end{lemma}
A direct implication of this lemma is the slow decay of $x$ with $\Lambda$ near the critical separatrix. To see this, we recall the general flow equation \eq{eq:x_flow} for $x$ 
\begin{equation}
    \frac{dx}{d\Lambda} = - \Lambda (1+x)x \frac{\rho^T_{\Lambda}(\Lambda)}{\ev{d}} \approx - \frac{\Lambda \rho^T_{\Lambda}(\Lambda)}{\ev{d}}x \,.
\end{equation}
Along the separatrix, Lemma~\ref{lemma:1pt_function_crit} implies $\ev{d} \sim \Lambda \log \Lambda$ and $\rho^T_{\Lambda}(\Lambda)$, consistent with the logarithmically slow growth rate of $x$. Below the separatrix, $\rho^T_{\Lambda}(\Lambda)$ becomes larger than $\rho^T_{\Lambda, \text{crit}}(\Lambda)$ as T-blocks are more likely to be decimated. But $\ev{d}^{-1}$ decays exponentially fast in $\int^{\Lambda} \rho^T_{\Lambda'}(\Lambda') d \Lambda'$. Therefore, the decay of $\ev{d}^{-1}$ overwhelms the growth of $\rho^T_{\Lambda}(\Lambda)$ and $\frac{dx}{d\Lambda}$ also becomes much smaller. As a result, the assumption of slow change in $x$ can be justified everywhere outside the thermal phase, a property that we will use repeatedly later.

Despite its power, this lemma only gives information about the marginal distributions precisely evaluated at the cutoff. A more complete understanding of the RG flow requires two additional assumptions:

\textit{Assumption 2: Along and below the critical separatrix, as the I-blocks become much longer than the T-blocks, the distribution of deficit lengths $\mu^I_{\Lambda}(d)$ for I-blocks tends to become wider and flatter. Concretely, we will assume that $\mu^I_{\Lambda}(\Lambda)$ decreases with $\Lambda$ and the derivative $\partial_d \mu^I_{\Lambda}(d)$ evaluated at the cutoff $d = \Lambda$ goes to zero sufficiently rapidly as $\Lambda \rightarrow \infty$:
\begin{equation}
    - \partial_d \log \mu^I_{\Lambda}(d), - \partial_d \log C^{TIT}_{\Lambda}(l_1,d,l_2)\big|_{d=\Lambda} \ll \rho^T_{\Lambda,\text{crit}}(\Lambda) \,.
\end{equation}}
What appears on the RHS is the probability of having T-blocks at cutoff along the critical separatrix. In the uncorrelated RG, $\mu^I_{\Lambda}(d)$ is an exponential distribution and one can easily show that
\begin{equation}
    - \partial_d \log \mu^I_{\Lambda}(d)\big|_{d=\Lambda} \sim \mu^I_{\Lambda}(\Lambda) \,.
\end{equation}
When we turn on correlations, we allow the distribution to change in form, but we relax the right hand side to $\rho^T_{\Lambda,\text{crit}}(\Lambda)$. This should be regarded as a weak assumption because T-block decimation always dominates along the separatrix, implying $\mu^I_{\Lambda,\text{crit}}(\Lambda) \ll \rho^T_{\Lambda,\text{crit}}(\Lambda)$. Below the separatrix, I-blocks at cutoff become even rarer and the inequality becomes more strongly satisfied.

\textit{Assumption 3: The connected distributions of nearest neighbor $T$ and $I$ blocks are upper-bounded by a function that's much larger than the product of marginal distributions for $l \in [\Lambda, \Lambda/x]$:}
\begin{equation}
    \left|\frac{C^{TI}_{\Lambda,c}(l, \Lambda)}{\rho^T_{\Lambda}(l) \mu^I_{\Lambda}(\Lambda)}\right|, \left|\frac{C^{TITI}_{\Lambda,c}(l_1, \Lambda,l_2,\Lambda)}{C^{TIT}_{\Lambda}(l_1,\Lambda,l_2) \mu^I_{\Lambda}(\Lambda)}\right| \ll \frac{\rho^T_{\Lambda,\text{crit}}(\Lambda)}{\mu^I_{\Lambda}(\Lambda)} \,.
\end{equation}
As we have argued in Sec.~\ref{subsec:stability_pos1}, correlations between nearby T and I-blocks should be washed out in the sense that joint moments $\ev{l_i^n d_i^m} \approx \ev{l_i^n} \ev{d_i^m}$. However, \emph{convergence of moments do not imply pointwise convergence of probability distributions}. Therefore, although it is reasonable to expect that connected correlators are suppressed relative to the product of marginals, we do not have a proof. To maximize the robustness of our arguments, we work with the weaker assumption 3 which allows the connected correlators to be much larger than the product of marginals pointwise but much smaller than the product of marginals multiplied by $\frac{\rho^T_{\Lambda,\text{crit}}(\Lambda)}{\mu^I_{\Lambda}(\Lambda)}$. Along the critical separatrix, this additional multiplicative factor diverges as a power law in $t \equiv \log \Lambda$. Below the separatrix, it diverges even faster since I-blocks at cutoff become stretched-exponentially rare. 

\indent With these weakened assumptions, the asymptotic projected flow equations for $x,y$ could be significantly modified. Nevertheless, the correlation length scaling doesn't change. To show this, we continue the analysis of $r_{\Lambda}(l)$ in \eq{eq:rLambda_l_flow}. Recall that $r_{\Lambda}(l) = \frac{x}{\ev{d}} \rho^T_{\Lambda}(l)$. Since the flow of $x$ is slow, the flow of $r_{\Lambda}(l)$ for $\Lambda \in [xl, l]$ is controlled by the competition between the growth of $\rho^T_{\Lambda}(l)$ and $\ev{d}$ with $\Lambda$. Clearly, the growth of both quantities is due to a monotonic decrease in the total number of blocks $N_{\Lambda}$. But for $\rho^T_{\Lambda}(l)$, there is an additional mechanism that reduces $\rho^T_{\Lambda}(l)$. This comes from decimations of T-blocks with length $l$ when a rare I-block is at cutoff. The rate of these processes is $C^{TI}_{\Lambda}(l,\Lambda)/\rho^T_{\Lambda}(l)$. Thus as long as $C^{TI}_{\Lambda}(l,\Lambda)/\rho^T_{\Lambda}(l) \ll \rho^T_{\Lambda}(\Lambda)$, which is the content of assumption 3, this decreasing contribution will be negligible and $r_{\Lambda}(l)$ will remain approximately constant for $\Lambda \in [xl, l]$. A more precise version of this argument in Appendix~\ref{sec:AppendixB} then leads to the key lemma:
\begin{lemma}\label{lemma:r_lambda_const}
    Under assumptions 1 and 3, along the separatrix we have $r_{\Lambda}(l) \approx r_l(l)$ up to errors of $\mathcal{O}(\log x^{-1} x^c)$ for all $\Lambda \in (xl, l)$ where $c >0$ and $x = x_{\Lambda}$. Below the separatrix, the error is strictly smaller, approaching $\mathcal{O}\left(\frac{1}{\rm{Superpoly}(\Lambda)}\right)$ in the large $\Lambda$ limit. 
\end{lemma}
The constancy of $r_{\Lambda}(l)$ for $\Lambda \in (xl , l)$, combined with the estimates in Lemma~\ref{lemma:1pt_function_flow}, allows us to compute the precise functional form of $\rho^T_{\Lambda}(l)$ along the critical separatrix. Importantly, $\rho^T_{\Lambda}(l)$ decays faster than $1/l^2$ everywhere below the separatrix, a property that will be used in the main argument.
\begin{lemma}\label{lemma:rhoT_scaling}
    $\rho^T_{\Lambda}(l) \approx \frac{x_{\Lambda}^{-1} \Lambda \log \Lambda}{x_l^{-1} l^2 \log l}$ for $l \in [\Lambda, 2\Lambda + \frac{\Lambda}{x}]$. This in turn implies that $\ev{l} \approx \Lambda \log x^{-1}$. 
\end{lemma}
With these technical lemmas in place, we can understand the flow equation for $r_{\Lambda}(l)$ \eq{eq:rLambda_l_flow} in the regime $\Lambda \in [xl, l]$ where the dominant growth mechanism for $\rho^T_{\Lambda}(l)$ is the production of new T-blocks with length $l$ from a $TIT \rightarrow T$ move. Following Ref.~\onlinecite{Morningstar_Huse_Imbrie_2020}, the strategy is to avoid solving the full integro-differential equation but instead derive an approximate recursion relation for $r_{\Lambda}(\frac{\Lambda}{x})$ in terms of $r_{\Lambda}(\Lambda)$. For that purpose, we fix $l = \frac{\Lambda}{x}$ and integrate \eq{eq:rLambda_l_flow} from $\Lambda' = x \Lambda$ to $\Lambda' = \Lambda$. Keeping all terms on the RHS for the moment, we find 
\begin{equation}\label{eq:rlambda_finite_flow}
    \begin{aligned}
        &r_{\Lambda'}(l)\big|_{x\Lambda}^{\Lambda} \approx \int_{x\Lambda}^{\Lambda} \bigg\{\left(-\frac{y}{\Lambda'} + \frac{\Lambda' \mu^I_{\Lambda'}(\Lambda')}{\ev{d}_{\Lambda'}} - \frac{2C^{TI}_{\Lambda'}(l, \Lambda')}{\rho^T_{\Lambda'}(l)}\right) \\
        &r_{\Lambda'}(l) + \frac{x}{\ev{d}_{\Lambda'}} \int_{\Lambda'}^{l - \Lambda' \frac{1+x}{x}} d l_1 C^{TIT}_{\Lambda'}(l_1, \Lambda', l-l_1 - \frac{\Lambda'}{x}) \bigg\} \,.
    \end{aligned}
\end{equation}
We will refer to the first line on the RHS as $F_{\Lambda, \text{depl}}$ and the second line as $F_{\Lambda, \text{prod}}$, in accordance with the general decomposition of flow equations into depletion and production terms in Appendix~\ref{sec:AppendixA}. As a sanity check, note that this reduces to the analogous flow equation (15) in Ref.~\onlinecite{Morningstar_Huse_Imbrie_2020} after we plug in the exact solution $\mu^I_{\Lambda}(d) = \mu^I_{\Lambda}(\Lambda) e^{-\mu^I_{\Lambda}(\Lambda)(d-\Lambda)}$ valid for uncorrelated disorder and factorize joint distributions into product of marginals. Now we make a change of variables from $\Lambda' \rightarrow l_2 = l-l_1 - \frac{\Lambda'}{x}$ to elucidate the physical picture. Using the slow decay of $x$, the production term $F_{\Lambda, \text{prod}}$ could be reduced to $x^2\int_D dl_1 dl_2 \frac{1}{\ev{d}_{\Lambda'}}C^{TIT}_{\Lambda'}(l_1, \Lambda', l_2)$ where $D$ is an isosceles triangular integration domain as shown in Fig.~\ref{fig:IntDomain} and $\Lambda' = \Lambda + 2 x (\Lambda - l_1 - l_2)$ depends implicitly on $l_1,l_2$. Roughly, this term counts all possible ways to form a T-block with length $l$ at scale $\Lambda$ by combining smaller blocks at an earlier stage with cutoff $\Lambda' \in [x\Lambda, \Lambda]$. When $l = \frac{\Lambda}{x}$, in order for the fractal structure of T-blocks to be dominant, this term should receive its dominant contributions from $\Lambda' = \Lambda$. Using the fundamental assumptions and Lemma~\ref{lemma:r_lambda_const} we can show that this is indeed the case. Moreover, the integral over domain $D$ can be replaced by an integral over an infinite region $[\Lambda, \infty]^2$ up to errors that are suppressed at large $\Lambda$ if the decay of $C^{TIT}_{\Lambda}(l_1,\Lambda,l_2)$ with $l_1,l_2$ is sufficiently fast. This is guaranteed by assumption 3 and the estimate in Lemma~\ref{lemma:rhoT_scaling}. Finally, since the integral over the infinite region $[\Lambda, \infty]^2$ simply gives $\mu^I_{\Lambda}(\Lambda)$, we immediately conclude
\begin{equation}\label{eq:F_Lambda_prod}
    F_{\Lambda, \rm prod} = \frac{x^2}{\ev{d}} \mu^I_{\Lambda}(\Lambda) \,.
\end{equation}
Using the decay properties of $\rho^T_{\Lambda}(l)$, we can also show that $F_{\Lambda,\text{depl}}$ is suppressed relative to $F_{\Lambda, \text{prod}}$, thereby establishing a recursion relation for $r_{\Lambda}(l)$
\begin{equation}
    r_{\Lambda/x}(\Lambda/x) \approx r_{\Lambda}(\Lambda/x) \approx \frac{x^2}{\ev{d}} \mu^I_{\Lambda}(\Lambda) \,.
\end{equation}
Using Lemma~\ref{lemma:r_lambda_const}, we complete the derivation of the projected flow equations in property (3)
\begin{equation}
    y_{\Lambda/x} = \frac{\Lambda^2}{x^2} r_{\Lambda/x}(\Lambda/x) \approx \frac{\Lambda^2}{\ev{d}} \mu^I_{\Lambda}(\Lambda) = \left(\frac{y_{\Lambda}}{x}\right)^2 \ev{d} \mu^I_{\Lambda}(\Lambda) \,.
\end{equation}
A byproduct of the precise argument in Appendix~\ref{sec:AppendixB} is that $\ev{d} \mu^I_{\Lambda}(\Lambda) \sim x \log \Lambda$ at large $\Lambda$ and along the separatrix. As $x \rightarrow 0$, by the assumption of critical slowing down, $x$ scales as a negative power of $\log \Lambda$. Thus $\ev{d} \mu^I_{\Lambda}(\Lambda) \sim x^c$ for some $c < 1$. From this we can simplify the recursion relation as
\begin{equation}
    y_{\Lambda/x} \sim \left(\frac{y_{\Lambda}}{x}\right)^2 x^c \approx \frac{y^2}{x^{2-c}} \,.
\end{equation}
This recursion is solved by $y \sim x^{\beta}$ with $\alpha = 2 - c$. Plugging this back into the exact flow equation \eq{eq:x_flow} for $x$ and introducing the RG time $t = \log \Lambda$, we can get the parametric form of the critical separatrix:
\begin{equation}
    \frac{dx}{dt} \approx - x^{\beta} \quad \rightarrow \quad x(t) \sim t^{\frac{1}{c-1}} \quad y(t) \sim t^{\frac{2-c}{c-1}} \,.
\end{equation}
When $\ev{d} \mu^I_{\Lambda}(\Lambda) \approx 1$ as in the uncorrelated RG, $c = 0$ and we recover the uncorrelated separatrix $x(t) \sim t^{-1}, y \sim t^{-2}$. In the correlated case, it is possible to have $c \neq 0$ so that $x,y$ have different scalings with $t$. Now let us consider deviations from the critical separatrix $y \approx x^{2-c + \delta_0}$. For $\delta_0$ sufficiently small, $\ev{d} \mu^I_{\Lambda}(\Lambda) \sim x^c$ continues to approximately hold. Thus for every recursion step, $t \rightarrow t + \log x^{-1}$,
\begin{equation}
    y_{\Lambda/x} \approx \left(\frac{x^{2-c + \delta_0}}{x}\right)^2 x^c \approx x^{2-c+2\delta_0} \,.
\end{equation}
The number of RG steps it takes for $\delta_0$ to reach an $\mathcal{O}(1)$ value is $\log_2 \delta_0^{-1}$. The elapsed RG time per RG step is $\frac{dt}{dn} = \log x^{-1} \approx \log (\log \Lambda)^{\frac{1}{1-c}} = \frac{1}{1-c} \log t$. This implies that the total RG time $T$ is related to the number of RG steps $N$ as:
\begin{equation}
     N \approx \int^T \frac{dt}{t'(n)} \approx \frac{(1-c)T}{\log T} \quad \rightarrow \quad T \approx \frac{N}{1-c} \log \frac{N}{1-c} \,.
\end{equation}
Hence in the limit $\delta_0 \rightarrow 0$, the correlation length scales as
\begin{equation}
    \xi = e^T \approx \delta_0^{-(1-c)^{-1} \log (\log_2 \delta_0^{-1}(1-c)^{-1})}  \,.
\end{equation}
As anticipated by the qualitative argument in Sec.~\ref{subsec:stability_pos2}, the above scaling satisfies $\nu = \infty$. The origin of the double logarithm is an extreme asymmetry between the logarithmic slowdown along the separatrix and the exponential speedup orthogonal to the separatrix. For any finite value of $c$, the double logarithmic scaling is robust up to $\delta_0$-independent constants. We have thus established property (3) in Sec.~\ref{subsec:overview}. 

Finally, to estimate the stretching exponent $\epsilon(x)$ deep in the MBL phase and obtain a more quantitative version of property (1), we have to study the scaling of $\ev{d} \mu^I_{\Lambda}(\Lambda)$. Using the flow equations for $\ev{d}$ and $\mu^I_{\Lambda}(\Lambda)$, 
\begin{equation}
    \frac{d \left[\ev{d} \mu^I_{\Lambda}(\Lambda)\right]}{d \Lambda} \approx - 2C^{TI}_{\Lambda,c}(\Lambda,\Lambda) \ev{d} + \ev{d} \partial_d \mu^I_{\Lambda}(d)\big|_{d=\Lambda} \,.
\end{equation}
The positivity of $C^{TI}_{\Lambda}$ and the numerical observation that $C^{TI}_{\Lambda,c}(\Lambda,\Lambda) < 0$ imply that $0 < -C^{TI}_{\Lambda,c}(\Lambda,\Lambda) <  \rho^T_{\Lambda}(\Lambda) \mu^I_{\Lambda}(\Lambda)$. As a result, we have a general scaling
\begin{equation}
    -2C^{TI}_{\Lambda,c}(\Lambda,\Lambda) \sim \eta \rho^T_{\Lambda}(\Lambda) \mu^I_{\Lambda}(\Lambda)\,,
\end{equation}
where $\eta < 1$ (if $\eta \geq 1$, then $\mu^I_{\Lambda}(\Lambda)$ would flow to $\infty$ as $\Lambda \rightarrow \infty$, which is impossible). This means that
\begin{equation}
    \frac{d \left[\ev{d} \mu^I_{\Lambda}(\Lambda)\right]}{d\Lambda} \approx \eta \rho^T_{\Lambda}(\Lambda) \left[\ev{d} \mu^I_{\Lambda}(\Lambda)\right]\,,
\end{equation}
which gives $\log \left[\ev{d} \mu^I_{\Lambda}(\Lambda)\right] \approx \eta \int^{\Lambda} \rho^T_{\Lambda'}(\Lambda') d \Lambda'$ upon integration. On the other hand, close to the MBL fixed line, $\rho^T_{\Lambda}(\Lambda) \gg \frac{1}{\Lambda}$, and $\log \ev{d} \approx \int^{\Lambda} \rho^T_{\Lambda'}(\Lambda') d \Lambda'$. By definition of $y$, we therefore conclude that $\log y \approx -\int^{\Lambda} \rho^T_{\Lambda'}(\Lambda') d \Lambda' +\, \text{subleading}$. Combining these estimates with the recursion relation, we find
\begin{equation}
    -\int^{\Lambda/x} \rho^T_{\Lambda'}(\Lambda') d \Lambda' = (\eta-2) \int^{\Lambda} \rho^T_{\Lambda'}(\Lambda') d \Lambda' - 2 \log x \,.
\end{equation}
Since $x$ freezes to a constant on the MBL fixed line, we can drop $2\log x$ as $\Lambda \rightarrow \infty$ in the above equation. The remaining equation is solved by the ansatz
\begin{equation}
    \rho^T_{\Lambda}(\Lambda) \sim \frac{1}{\Lambda^{1 - \epsilon}}, \quad  - \left(\frac{\Lambda}{x}\right)^{\epsilon} = (-2+\eta) \Lambda^{\epsilon} \,.
\end{equation}
We now recognize $\epsilon$ as the stretching exponent, which must satisfy
\begin{equation}
    \epsilon = \frac{\log (2 - \eta)}{\log x^{-1}}, \quad 0 \leq \eta < 1 \,.
\end{equation}
Though this result holds for general $\eta$, the moment bounds in \eq{eq:moment_bound} strongly suggest that $\eta \rightarrow 0$ as $\Lambda \rightarrow \infty$. Therefore, the expectation is that $\epsilon = \frac{\log 2}{\log x^{-1}}$, which is exactly the uncorrelated value. At this point all three features promised in Sec.~\ref{subsec:overview} have been established. 

\section{Discussion}\label{sec:discussion}

In this work, we have argued via an analytic renormalization group approach that the Morningstar-Huse-Imbrie critical scaling for the MBL transition is not affected by the introduction of spatial correlations in the distribution of initial block parameters. Since our argument applies to correlations with arbitrary wandering exponents, it provides strong evidence that the MHI critical scaling is in fact a robust universality class. However, a few challenges and open questions remain unresolved. 

At the most basic level, our analytic arguments in the case of positive correlations require three technical assumptions. Though all three assumptions are motivated by physical arguments/supported by finite-size numerics, it would be more satisfactory to deduce them from a more fundamental principle. But we believe this is merely an aesthetic issue. A more serious concern, even granted the correctness of the assumptions, is the robustness of our conclusion beyond this toy RG. After all, the MHI RG rules reduce the complicated phenomenology of MBL into alternating T and I-blocks characterized only by a few block parameters. One can easily imagine that tweaking the choice of block parameters and RG rules could give rise to a totally different universality class. Therefore, it would be much more satisfying to derive the coarse-grained flow equations in the MHI RG without committing to a specific microscopic RG. 

As shown in Sec.~\ref{sec:stability_hyp}, the irrelevance of hyperuniform correlations is indeed a universal feature of all asymptotically additive RGs. In contrast, arguments for the irrelevance of positive correlations in Sec.~\ref{subsec:stability_pos1} and Sec.~\ref{subsec:stability_pos2} relied on specific features of the MHI RG. Nevertheless, as we explained in these sections, the key physical input that renders positive correlations irrelevant is the asymmetric thermalizing power of T and I-blocks, which is a universal consequence of the avalanche mechanism. Therefore, we anticipate that a more robust argument for universality will involve avalanches in an essential way. A first step in that direction was attempted in Ref.~\onlinecite{Dumitrescu_Goremykina_Parameswaran_Serbyn_Vasseur_2019}, where the avalanche mechanism combined with an assumption on the analyticity of the RG $\beta$-functions led to the KT universality class on general grounds. But this conclusion was called into question by the MHI RG, which, despite its simple and well-motivated microscopic rules, featured a non-analytic $\beta$-function in the $(x,y)$ plane. Whether such non-analyticities should be expected in general is a question that can hopefully be settled by improving the arguments of Ref.~\onlinecite{Dumitrescu_Goremykina_Parameswaran_Serbyn_Vasseur_2019}. An alternative possibility is that analyticity only holds in a higher dimensional parameter space and fails when we project onto the two dimensional subspace spanned by $(x,y)$. The challenge would then be to identify the minimal set of parameters needed for analyticity and constrain the form of the $\beta$-function in this bigger space using some general physical principles (with quantum avalanche probably playing a key role). Such a framework would obviate the need for more microscopic RG models and provide a much more robust picture of the MBL transition.

Even if we assume that MHI RG rules indeed capture the correct universal properties of the MBL transition, it remains to be understood whether there exists a class of initial disorder correlation that would give rise to a different critical scaling. One interesting example is the quasiperiodic MBL transition~\cite{PhysRevLett.119.075702,Agrawal_Gopalakrishnan_Vasseur_2020_QP}. For the case of the simplified RG in Ref.~\onlinecite{Zhang_Zhao_Devakul_Huse_2016}, which assumed some symmetry between the MBL and ETH phases, quasiperiodic initial conditions give rise to a critical exponent $\nu = 1$~\cite{Agrawal_Gopalakrishnan_Vasseur_2020_QP}, which is distinct from the value of $\nu$ for the uncorrelated random case and for the two families of correlations that we have considered (see Appendix~\ref{sec:AppendixC} for a more detailed discussion of the symmetric RG). In the MHI RG, the analytic framework that we have developed for hyperuniform and positive correlations does not apply to quasiperiodic correlations which are described by a fixed set of initial block lengths rather than an ensemble. Therefore, although the wandering exponent of quasiperiodic correlation $w = 0$ coincides with that of hyperuniform correlation with $\alpha = 1$, we cannot conclude that $\nu = \infty$ for the quasiperiodic case. Understanding the fate of the quasiperiodic MBL transition will likely require new analytic insights. 

Finally, we should remark that the moment bounds of joint distributions $C^{TI\ldots T}_{\Lambda}$ in block RGs and the BBGKY hierarchy of correlated flows developed in this paper may have applications to more general functional RGs. One topic where the formalism might be helpful is the generalized Harris bound/correlated CCFS bound that we discussed in the introduction. 

\noindent \emph{Acknowledgments.---}\, We thank U. Agrawal, D. Huse, A. Morningstar, H. Singh, M. Serbyn and B. Ware for useful discussions and collaboration on related works. We also thank Yunkun Zhou for helpful comments on bounding probability distributions. We acknowledge support from NSF Grants No. DMR-2103938 (S.G.), DMR-2104141 (R.V.), the US Department of Energy, Office of Science, Basic Energy Sciences, under Early Career Award No. DE-SC0021111 (V.K.), the Alfred P. Sloan Foundation through Sloan Research Fellowships (V.K. and R.V.), and the Packard Fellowship in Science and Engineering (V.K.).  We also acknowledge the Sherlock High Performance Computing Cluster at Stanford for providing computing resources.

\newpage 
\begin{widetext}

\appendix 

\section{Derivation of correlated flow equations}\label{sec:AppendixA}

As pointed out in Sec.~\ref{subsec:overview}, spatial correlations in the initial T/I-block lengths force us to consider a functional RG of the joint probability distribution $P_{\Lambda}(\vec l^T, \vec d)$ instead of the single block marginals $\rho^T_{\Lambda}(l^T), \mu^I_{\Lambda}(d)$ studied in~\onlinecite{Morningstar_Huse_Imbrie_2020}. However, to understand the behavior of the order parameters $x_{\Lambda}, y_{\Lambda}$ which depend only on single block marginals and their integrals, we do not need to keep track of the renormalization of the full probability distribution $P_{\Lambda}(\vec l^T, \vec d)$. Instead, we will develop a hierarchy of equations (analogous to the BBGKY hierarchy is classical statistical mechanics), where the marginal distribution of n nearest neighbors depends on the distribution of n+2 nearest neighbors. This set of equations do not close at any finite order but will be sufficient for a precise analysis of the near-critical regime. 

The general structure of the flow equations can be understood via the following schematic equation:
\begin{equation}
    \partial_{\Lambda} C^{(n)}_{\Lambda}(l_1,d_1,\ldots) = F_{\rm depl}[C^{(n+2)}_{\Lambda}] + F_{\rm prod}[C^{(n+2)}_{\Lambda}] \,.
\end{equation}
Here, $C^{(n)}_{\Lambda}(l_1,d_1,\ldots)$ is the probability that a contiguous chain of n blocks have lengths $(l_1,d_1,\ldots)$ when the length cutoff is $\Lambda$ (note that we have dropped the superscript $T$ for notational convenience). The flow of $C^{(n)}_{\Lambda}$ has two contributions: $F_{\text{depl}}$ is a depletion term that accounts for RG moves where a chain of n blocks with length $(l_1,d_1,\ldots)$ at cutoff $\Lambda$ are destroyed when we move the cutoff $\Lambda + d \Lambda$; in contrast, $F_{\text{prod}}$ is a production term that encodes RG moves that create a new chain of n blocks with lengths $(l_1,d_1,\ldots)$ which was not present at cutoff $\Lambda$. To be concrete, we consider the case where n = 1, so that the LHS is just a single-block marginal. For T and I-blocks we have
\begin{equation}\label{1pt_flow_appendix}
    \begin{aligned}
        \partial_{\Lambda} \rho^T_{\Lambda}(l) &= \rho^T_{\Lambda}(l)\left[\mu^I_{\Lambda}(\Lambda) + \rho^T_{\Lambda}(\Lambda)\right] - 2C^{TI}_{\Lambda}(l, \Lambda) + \int_{\Lambda}^{l - \frac{\Lambda}{x} - \Lambda} C^{TIT}_{\Lambda,c}\left(l_1, \Lambda, l - \frac{\Lambda}{x} - l_1\right) d l_1 \,,
    \end{aligned}
\end{equation}
\begin{equation}
    \begin{aligned}
        \partial_{\Lambda} \mu^I_{\Lambda}(d) &= \mu^I_{\Lambda}(d)\left[\mu^I_{\Lambda}(\Lambda) + \rho^T_{\Lambda}(\Lambda)\right] - 2C^{TI}_{\Lambda}(\Lambda, d) + \int_{\Lambda}^{d} C^{ITI}_{\Lambda}\left(s, \Lambda, d + \Lambda - s\right) d s  \,.
    \end{aligned}
\end{equation}
In each of the flow equations above, the first and second lines are depletion and production terms respectively. In the uncorrelated limit, the multi-block correlations factorize as products of single-block marginals $C^{TI}(l,d) = \rho^T_{\Lambda}(l) \mu^I_{\Lambda}(d), C^{TIT}(l_1,d,l_2) = \rho^T(l_1) \mu^I(d) \rho^T(l_2)$ etc. and the resulting flow equations agree with those obtained in~\onlinecite{Morningstar_Huse_Imbrie_2020}. Following the same procedure, it is conceptually simple to derive the full hierarchy of equations. For clarity, we will only present the next simplest equation in the hierarchy (which turns out to be all we need near criticality). This flow equation will involve two types of contributions: one coming from the depletion of thermal blocks with length $l_1,l_2$ that have already been produced in earlier stages in the RG (i.e. at cutoff smaller than $\Lambda$) and the other coming from the production of thermal blocks with length $l_1,l_2$ due to decimation of I-blocks with $d = \Lambda$. Taking both contributions into account, the number density flows as
\begin{equation}
    \begin{aligned}
    n^{TIT}_{\Lambda+d \Lambda}(l_1,\Lambda+d \Lambda, l_2) &= n^{TIT}_{\Lambda}(l_1, \Lambda, l_2) \\
    &+ d \Lambda \bigg[- C^{ITIT}_{\Lambda}(\Lambda,l_1,d,l_2) - C^{TITI}_{\Lambda}(l_1,d,l_2,\Lambda) + \int_{\Lambda}^{\infty} d \tilde l_1 C^{TITIT}_{\Lambda}\left(\tilde l_1, \Lambda, l_1 - \frac{\Lambda}{x} - \tilde l_1, d, l_2\right) \\
    &+ \int_{\Lambda}^{\infty} d \tilde l_1 C^{TITIT}_{\Lambda}\left(l_1, d, \tilde l_2, \Lambda, l_2 - \frac{\Lambda}{x} - \tilde l_2\right) + \int_{\Lambda}^{\infty} ds C^{TITIT}_{\Lambda}(l_1, s, \Lambda, d+\Lambda -s, l_2) \bigg] \,.
    \end{aligned}
\end{equation}
Let the total number of I-blocks be $N_{\Lambda}$ (which is equal to the number of T-blocks) with flow equation $N_{\Lambda+d\Lambda} = N_{\Lambda} \left[1 - d \Lambda (\rho^T_{\Lambda}(\Lambda) + \mu^I_{\Lambda}(\Lambda)\right]$. Dividing both sides of the previous equation by $N_{\Lambda+d\Lambda}$ and recalling the definition $n^{TIT}_{\Lambda}/N_{\Lambda} = C^{TIT}_{\Lambda}$, we get
\begin{equation}\label{eq:3_neighbor_corr_flow}
    \partial_{\Lambda} C^{TIT}_{\Lambda}(l_1,d,l_2) = F_{\rm depl} + F_{\rm prod} \,,
\end{equation}
where the depletion and production terms are given by
\begin{equation}
    \begin{aligned}
    F_{\rm depl} &= C^{TIT}_{\Lambda}(l_1,d,l_2)\left[\rho^T_{\Lambda}(\Lambda) + \mu^I_{\Lambda}(\Lambda)\right] - C^{ITIT}_{\Lambda}(\Lambda,l_1,d,l_2) - C^{TITI}_{\Lambda}(l_1,d,l_2,\Lambda) \,,
    \end{aligned}
\end{equation}
\begin{equation}
    \begin{aligned}
    F_{\rm prod} &= \int_{\Lambda}^{\infty} d s \left[C^{TITIT}_{\Lambda}\left(s, \Lambda, l_1 - \frac{\Lambda}{x} - s, d, l_2\right) + C^{TITIT}_{\Lambda}\left(l_1, d, s, \Lambda, l_2 - \frac{\Lambda}{x} - s\right) + C^{TITIT}_{\Lambda}(l_1, s, \Lambda, d+\Lambda -s, l_2) \right] \,.
    \end{aligned}
\end{equation}
The two-parameter RG flows that we are interested in depend not on the full distribution but rather on the moments of these distributions. It is easy to derive flow equations for the moments from the full PDEs above. We enumerate a full set of these equations here:
\begin{align}
    \frac{d \ev{l^I}}{d\Lambda} &= \mu^I_{\Lambda}(\Lambda)\left[\ev{l^I} - \Lambda/x\right] + \rho^T_{\Lambda}(\Lambda) \left[\ev{l^I} + \Lambda\right] \,,\\
    \frac{d \ev{l^T}}{d\Lambda} &= \rho^T_{\Lambda}(\Lambda)\left[\ev{l^T} - \Lambda\right] + \mu^I_{\Lambda}(\Lambda) \left[\ev{l^T} + \Lambda/x\right]\,, \\
    \frac{d \ev{d}}{d\Lambda} &= \left[\mu^I_{\Lambda}(\Lambda) + \rho^T_{\Lambda}(\Lambda)\right](\ev{d} - \Lambda) \,.
\end{align}
Note that these equations are equivalent in form to their uncorrelated analogues in~\onlinecite{Morningstar_Huse_2019}. However, since $\rho^T_{\Lambda}(\Lambda), \mu^I_{\Lambda}(\Lambda)$ are shifted by correlations, the functional dependence of $\ev{d}, \ev{l^I}, \ev{l^T}$ on $\Lambda$ may also be different from the uncorrelated RG. 

The proof of these three formulae are structurally identical and rather tedious, so we will only give the simplest version of the calculation. Consider $\ev{d} = \int_{\Lambda}^{\infty} dd \mu^I_{\Lambda}(d) \cdot d$. A few simple algebraic manipulations lead to 
\begin{equation}
    \begin{aligned}
        \frac{d}{d\Lambda} \ev{d} &= - \mu^I_{\Lambda}(\Lambda) \Lambda + \int_{\Lambda}^{\infty} \partial_{\Lambda} \mu^I_{\Lambda}(d) \cdot d \\
        &= - \mu^I_{\Lambda}(\Lambda) \Lambda + \int_{\Lambda}^{\infty} \mu^I_{\Lambda}(d) \left[\mu^I_{\Lambda}(\Lambda) + \rho^T_{\Lambda}(\Lambda)\right] d - 2\int_{\Lambda}^{\infty} C^{TI}_{\Lambda}(\Lambda, d) d + \int_{\Lambda}^{\infty} d d \int_{\Lambda}^d dx C^{ITI}_{\Lambda}(d-x+\Lambda, \Lambda, x) d \\
        &= - \mu^I_{\Lambda}(\Lambda) \Lambda + \ev{d} \left[\mu^I_{\Lambda}(\Lambda) + \rho^T_{\Lambda}(\Lambda)\right] - 2 \ev{d_i|l^T_i = \Lambda} + \int_{\Lambda}^{\infty} dx \int_{\Lambda}^{\infty} dy \int d d \delta(d-x+\Lambda - y) C^{ITI}_{\Lambda}(y,\Lambda,x) d \\
        &= - \mu^I_{\Lambda}(\Lambda) \Lambda + \ev{d} \left[\mu^I_{\Lambda}(\Lambda) + \rho^T_{\Lambda}(\Lambda)\right] - 2 \ev{d_i|l^T_i = \Lambda} + 2 \ev{d_i|l^T_i = \Lambda} - \Lambda \rho^T_{\Lambda}(\Lambda) \\
        &= \left[\mu^I_{\Lambda}(\Lambda) + \rho^T_{\Lambda}(\Lambda)\right] (\ev{d} - \Lambda)  \,.
    \end{aligned}
\end{equation}
where in the third line we introduced the conditional expectation value $\ev{\cdot|\cdot}$ and in the fourth line we used the following facts that follow from definition:
\begin{equation}
    \int C^{ITI}_{\Lambda}(x,l,y) dy = C^{TI}_{\Lambda}(l,x) \quad \int C^{ITI}_{\Lambda}(x,l,y) dx = C^{TI}_{\Lambda}(l,y) \quad \int C^{ITI}_{\Lambda}(x,l,y) dx dy = \rho^T_{\Lambda}(l) \,.
\end{equation}
This completes the proof of Lemma~\ref{lemma:1pt_function_flow}.

\section{Derivation of recursion relation for correlated MHI RG}\label{sec:AppendixB}

In the main text, we outlined an argument for the stability of MBL criticality in the MHI RG. In this appendix, we will fill in some of the technical details. Recall the fundamental variables in our RG:
\begin{equation}
    x_{\Lambda} = \frac{\ev{d}}{\ev{l^I}} \,,\quad r_{\Lambda}(l) = \frac{x}{\ev{d}} \rho^T_{\Lambda}(l) \,,\quad y = r_{\Lambda}(\Lambda) \Lambda^2 \,.
\end{equation}
In Sec.~\ref{subsec:stability_pos3}, we already derived the exact flow equation $\frac{dx}{d\Lambda} = - \frac{(1+x) y}{\Lambda}$ for $x_{\Lambda}$. Here we will derive the flow equation for $r_{\Lambda}(l)$ which will facilitate our analysis of $y$ later. Differentiating the definition directly and using the exact flow equation for $x$, we get three terms which partially cancel each other:
\begin{equation}
    \begin{aligned}
    \partial_{\Lambda} r_{\Lambda}(l) &= \frac{dx}{d\Lambda} \frac{\rho^T_{\Lambda}(l)}{\ev{d}} - \frac{x}{\ev{d}^2} \rho^T_{\Lambda}(l) \partial_{\Lambda} \ev{d} + \frac{x}{\ev{d}} \partial_{\Lambda} \rho^T_{\Lambda}(l) \\
    &= \frac{1}{x} \frac{dx}{d\Lambda} r_{\Lambda}(l) - \frac{\ev{d} - \Lambda}{\ev{d}} \left[\mu^I_{\Lambda}(\Lambda) + \rho^T_{\Lambda}(\Lambda)\right] r_{\Lambda}(l) \\
    &\hspace{0.5cm} + \frac{x}{\ev{d}} \big\{\rho^T_{\Lambda}(l) \left[\mu^I_{\Lambda}(\Lambda) + \rho^T_{\Lambda}(\Lambda)\right] - 2 C^{TI}_{\Lambda}(l,\Lambda) + \int_{\Lambda}^{l - \Lambda(1+x^{-1})} C^{TIT}_{\Lambda}\left(l_1,\Lambda, l-\frac{\Lambda}{x} - l_1\right) d l_1 \big\} \\
    &= \big(- \frac{y}{\Lambda} + \frac{\Lambda \mu^I_{\Lambda}(\Lambda)}{\ev{d}} - \frac{2 C^{TI}_{\Lambda}(l,\Lambda)}{\rho^T_{\Lambda}(l)} \big) r_{\Lambda}(l) + \frac{x}{\ev{d}} \int_{\Lambda}^{l - \Lambda(1+x^{-1})} C^{TIT}_{\Lambda}\left(l_1,\Lambda, l-\frac{\Lambda}{x} - l_1\right) d l_1 \,.
    \end{aligned} 
\end{equation}
Integrating the above equation from $x\Lambda$ to $\Lambda$ gives equation \eq{eq:rlambda_finite_flow} in the main text, which we rewrite below
\begin{equation}\label{eq:rlambda_finite_flow_appendix}
    \begin{aligned}
        r_{\Lambda'}(l)|_{x\Lambda}^{\Lambda} &\approx \int_{x\Lambda}^{\Lambda} \bigg\{\left(-\frac{y}{\Lambda'} + \frac{\Lambda' \mu^I_{\Lambda'}(\Lambda')}{\ev{d}_{\Lambda'}} - \frac{2C^{TI}_{\Lambda'}(l, \Lambda')}{\rho^T_{\Lambda'}(l)}\right) r_{\Lambda'}(l) + \frac{x}{\ev{d}} \int_{\Lambda'}^{l - \Lambda'(1+x^{-1})} C^{TIT}_{\Lambda'}\left(l_1,\Lambda', l-\frac{\Lambda'}{x} - l_1\right) d l_1 \bigg\} \,.
    \end{aligned}
\end{equation}
As explained in the main text, our goal is to show that the second term involving $C^{TIT}$ dominates over the first term. To do that, we first recall the \textbf{fundamental assumptions} that go into our analysis: 

\textit{Assumption 1: Critical slowing down holds along the separatrix so that $x \sim t^{-\alpha}$ for some positive exponent $\alpha$ where $t = \log \Lambda$ is the RG time. In the uncorrelated MHI RG, $\alpha = 1$. Here we only assume that $\alpha$ is finite. }

\textit{Assumption 2: Along and below the critical separatrix, as the I-blocks become much longer than the T-blocks, the distribution of deficit lengths $\mu^I_{\Lambda}(d)$ for I-blocks tends to become wider and flatter. Concretely, we will assume that $\mu^I_{\Lambda}(\Lambda)$ decreases with $\Lambda$ and the derivative $\partial_d \mu^I_{\Lambda}(d)$ evaluated at the cutoff $d = \Lambda$ goes to zero sufficiently rapidly as $\Lambda \rightarrow \infty$:
\begin{equation}
    - \partial_d \log \mu^I_{\Lambda}(d), - \partial_d \log C^{TIT}_{\Lambda}(l_1,d,l_2)\big|_{d=\Lambda} \ll \rho^T_{\Lambda,\rm crit}(\Lambda) \,.
\end{equation}}

\textit{Assumption 3: The connected distributions of nearest neighbor $T$ and $I$ blocks are upper-bounded by a function that's much larger than the product of marginal distributions for $l \in [\Lambda, \Lambda/x]$:
\begin{equation}
    \left|\frac{C^{TI}_{\Lambda,c}(l, \Lambda)}{\rho^T_{\Lambda}(l) \mu^I_{\Lambda}(\Lambda)}\right|, \left|\frac{C^{TITI}_{\Lambda,c}(l_1, \Lambda,l_2,\Lambda)}{C^{TIT}_{\Lambda}(l_1,\Lambda,l_2) \mu^I_{\Lambda}(\Lambda)}\right| \ll \frac{\rho^T_{\Lambda,\rm crit}(\Lambda)}{\mu^I_{\Lambda}(\Lambda)} \,.
\end{equation}
where $\ll$ means a multiplicative factor that's $\frac{1}{\text{Poly}(t)}$ and hence only logarithmically small in $\Lambda$.}

The motivation and numerical tests for these assumptions are discussed in Sec.~\ref{subsec:stability_pos3} and will not be repeated here. With these assumptions in mind, we continue the analysis of \eq{eq:rlambda_finite_flow_appendix}. We denote the first term on the RHS of \eq{eq:rlambda_finite_flow_appendix} by $F_{\Lambda, \text{depl}}$ since it captures the depletion of T-blocks that have been produced at an earlier RG time. The second term which captures the production of new T-blocks will be denoted by $F_{\Lambda, \text{prod}}$. As a general observation, note that the various terms in $F_{\Lambda,\text{depl}}, F_{\Lambda,\text{prod}}$ correspond to various levels of the hierarchy: from the lowest level moments $\ev{\ldots}$ to the highest level three-block joint distributions. The flow equations are a bridge between low and high levels in the hierarchy. The general structure of the two-parameter phase diagram gives us important information about low levels in the hierarchy, for example the relationship between various moments $\ev{l^T}, \ev{l^I}, \ev{d}$. So our general strategy is to bootstrap the higher level distributions from the low level information via the flow equations. 
\begin{lemma}\label{lemma1:appendix}
    Under assumption 1, $\ev{d} \sim \Lambda \log \Lambda \gg \Lambda$, $\mu^I_{\Lambda}(\Lambda) < x\rho^T_{\Lambda}(\Lambda)$, $\rho^T_{\Lambda} \approx \frac{1}{\Lambda} + \mathcal{O}(1)\frac{1}{\Lambda \rm{Poly}(\log \Lambda)})$ in the large $\Lambda$ limit along the critical separatrix. Below the critical separatrix (in the MBL phase), $\ev{d} \gg \Lambda$ and $\mu^I_{\Lambda}(\Lambda) < x \rho^T_{\Lambda}(\Lambda)$ still hold.
\end{lemma}
\begin{proof}
    We begin by showing that $\ev{d} \gg \Lambda$ without fixing its precise scaling. Suppose that $\ev{d} \gg \Lambda$ doesn't hold, then since $\ev{d} > \Lambda$, there must be a finite $k \geq 1$ such that $\ev{d} \sim k \Lambda + \text{subleading}$ at large $\Lambda$. The flow equation for $\ev{d}$ thus reduces to 
    \begin{equation}
        k = \frac{d \ev{d}}{d\Lambda} = \left[\rho^T_{\Lambda}(\Lambda) + \mu^I_{\Lambda}(\Lambda)\right] (\ev{d} - \Lambda) = \left[\rho^T_{\Lambda}(\Lambda) + \mu^I_{\Lambda}(\Lambda)\right] \left[(k-1) \Lambda + o(\Lambda)\right] \,.
    \end{equation}
    If $k > 1$, then to leading order in $\Lambda$, $\mu^I_{\Lambda}(\Lambda) + \rho^T_{\Lambda}(\Lambda) \approx \frac{k}{k-1} \frac{1}{\Lambda}$. By combining the flow equation for $\ev{l^I}$ and $\ev{l^T}$, we immediately see that $\ev{l^I} + \ev{l^T} \sim \Lambda^{1 + \frac{1}{k-1}}$. As argued in the main text (and suggested by numerics) the thermal fraction flows to zero along the critical separatrix, implying that $\ev{l^I} > \ev{l^T}$. This forces the scaling $\ev{l^I} \sim \Lambda^{1+\frac{1}{k-1}}$and hence $x \sim \Lambda^{-\frac{1}{k-1}}$. For any $k > 1$, this is an exponentially fast decay in the RG time $t = \log \Lambda$, violating assumption 1. If $k = 1$, then $\mu^I_{\Lambda}(\Lambda) + \rho^T_{\Lambda}(\Lambda) \approx \frac{1}{\mathcal{O}(1)\Lambda)}$, which implies that $\ev{l^I} + \ev{l^T}$ grows faster than any power law in $\Lambda$ and $x$ decays faster than any power law in $\Lambda$. This leads to an even stronger violation of assumption 1. Hence, we conclude that $\ev{d} \gg \Lambda$. 
    
    To learn something about $\mu^I_{\Lambda}(\Lambda), \rho^T_{\Lambda}(\Lambda)$, we go back to the flow equation for $\ev{l^I}, \ev{l^T}$. Since $\ev{l^T}/\ev{l^I} \rightarrow 0$ \textbf{along and below} the critical separatrix, we must demand $\log f(\Lambda)$ to be a monotonically decreasing function at large $\Lambda$ where $f$ is the thermal fraction. Using the flow equation for $\ev{l^T}, \ev{l^I}$, it is easy to see that
    \begin{equation}
        \frac{d \log f}{d \Lambda} = \frac{\Lambda}{\ev{l^T}} \left[\frac{\mu^I_{\Lambda}(\Lambda)}{x} - \rho^T_{\Lambda}(\Lambda)\right] \,.
    \end{equation}
    In order for the RHS to be negative, we must have $\mu^I_{\Lambda}(\Lambda) < x \rho^T_{\Lambda}(\Lambda)$ as claimed in the lemma. 
    
    To further fix the precise scaling of $\rho^T_{\Lambda}(\Lambda)$ and $\ev{d}$, we go back to the flow equation for $x$:
    \begin{equation}
        \frac{dx}{d\Lambda} = - \frac{\Lambda^2 x}{\Lambda \ev{d}} \rho^T_{\Lambda}(\Lambda) = - \frac{\Lambda x \rho^T_{\Lambda}(\Lambda)}{\ev{d}} \quad \rightarrow \quad \frac{d \log x}{d \log \Lambda} = -\frac{\Lambda^2 \rho^T_{\Lambda}(\Lambda)}{\ev{d}} \,.
    \end{equation}
    In the large $\Lambda$ limit, we have already shown that $\rho^T_{\Lambda}(\Lambda) \gg \mu^I_{\Lambda}(\Lambda)$ and $\ev{d} \gg \Lambda$. Therefore, the flow equation for $\ev{d}$ reduces to $\frac{d \ev{d}}{d\Lambda} = [\rho^T_{\Lambda}(\Lambda) + \text{subleading}] (\ev{d} - \Lambda)$. We now do a simple case work. If $\lim_{\Lambda \rightarrow \infty} \Lambda \rho^T_{\Lambda}(\Lambda) > 1$, then $\ev{d}$ grows at least as fast as $\Lambda^{1+\epsilon}$ for some $\epsilon > 0$. This would mean $\frac{d \log x}{d \log \Lambda} \sim - \frac{1}{\Lambda^{\epsilon}}$ which implies that $x \sim \frac{1}{\text{Poly}(\Lambda)}$, leading to a contradiction. On the other hand, if $\lim_{\Lambda \rightarrow \infty} \Lambda \rho^T_{\Lambda}(\Lambda) < 1$, then $\ev{d} = o(\Lambda)$, violating the constraint that $\ev{d} > \Lambda$. Hence, we conclude that $\rho^T_{\Lambda}(\Lambda) \approx \frac{1}{\Lambda} + \text{subleading}$. To get the scaling of the subleading terms, we return again to the flow equation for $x$. By assumption, $x$ decays as a polynomial in the RG time $\log \Lambda$. This implies that $\frac{dx}{d\Lambda} \sim -\frac{x}{\Lambda \log \Lambda}$, which, together with the leading order scaling of $\rho^T_{\Lambda}(\Lambda)$, forces $\ev{d} \sim \Lambda \log \Lambda$. In order for this scaling to be consistent with $\rho^T_{\Lambda}(\Lambda) \approx \frac{1}{\Lambda}, \mu^I_{\Lambda}(\Lambda) < x \rho^T_{\Lambda}$ and the exact flow equation $\frac{d \ev{d}}{d\Lambda} = \left[\rho^T_{\Lambda}(\Lambda) + \mu^I_{\Lambda}(\Lambda)\right] (\ev{d} - \Lambda)$, we must demand $\rho^T_{\Lambda}(\Lambda) + \mu^I_{\Lambda}(\Lambda) \approx \frac{1}{\Lambda} + o(\frac{1}{\Lambda \log \Lambda})$. Independent of the precise scaling of $x$ with $\Lambda$, we can conclude that $\rho^T_{\Lambda}(\Lambda) \approx \frac{1}{\Lambda} + \mathcal{O}(1)\frac{1}{\Lambda \text{Poly}(\log \Lambda)})$. These are all the estimates available along the separatrix.
    
    Below the separatrix, $\rho^T_{\Lambda}(\Lambda) > \frac{1}{\Lambda}$ because T-blocks become even likely to be decimated. This means that $\ev{d} \sim e^{\int^{\Lambda} \rho^T_{\Lambda'}(\Lambda') d \Lambda'}$ grows faster than along the separatrix and hence faster than $\Lambda$. This concludes the second part of the Lemma. 
\end{proof}
\begin{lemma}\label{lemma2:appendix}
    Under assumptions 1 and 3, along the separatrix we have $r_{\Lambda}(l) \approx r_l(l)$ up to errors of $\mathcal{O}(1)\log x^{-1} x^c)$ for all $\Lambda \in (xl, l)$ where $c > 0$ and $x = x_{\Lambda}$. Below the separatrix, the error is strictly smaller, approaching $\mathcal{O}\left(\frac{1}{\rm{Superpoly}(\Lambda)}\right)$ in the large $\Lambda$ limit. 
\end{lemma}
\begin{proof}
    In the flow equation \eq{eq:rlambda_finite_flow_appendix}, production terms give zero contribution whenever $\Lambda \in (xl,l)$ because the shortest T-block created from a $TIT \rightarrow T$ move at cutoff $\Lambda$ has length greater than or equal to $\frac{\Lambda}{x}$. Hence we may focus only on the depletion term
    \begin{equation}
        \partial_{\Lambda} r_{\Lambda}(l) = \left(- \frac{y}{\Lambda} + \frac{\Lambda \mu^I_{\Lambda}(\Lambda)}{\ev{d}} - \frac{2 C^{TI}_{\Lambda}(l,\Lambda)}{\rho^T_{\Lambda}(l)}\right) r_{\Lambda}(l) \quad \rightarrow \quad r_{\Lambda_2}(l) = e^{\int_{\Lambda_1}^{\Lambda_2} \left(- \frac{y}{\Lambda} + \frac{\Lambda \mu^I_{\Lambda}(\Lambda)}{\ev{d}} - \frac{2 C^{TI}_{\Lambda}(l,\Lambda)}{\rho^T_{\Lambda}(l)}\right) d \Lambda} r_{\Lambda_1}(l) \,.
    \end{equation}
    We first work along the separatrix. At large $\Lambda$, $y = \frac{x}{\ev{d}} \rho^T_{\Lambda}(\Lambda) \Lambda^2 \sim \frac{x \Lambda}{\ev{d}}$. Since $x \ll 1$, $\frac{y}{\Lambda} = \frac{x}{\ev{d}} \sim \frac{x}{\Lambda \log \Lambda}$. Similarly, as $\mu^I_{\Lambda}(\Lambda) \approx \frac{x}{\Lambda}$, the second term $\frac{\Lambda \mu^I_{\Lambda}(\Lambda)}{\ev{d}} \approx \frac{x}{\ev{d}} \approx \frac{x}{\Lambda \log \Lambda}$. Finally, by assumption 3, $C^{TI}_{\Lambda}(l,\Lambda)/\rho^T_{\Lambda}(l) \ll \rho^T_{\Lambda}(\Lambda) \sim \frac{1}{\Lambda}$. Here $\ll$ means a multiplicative factor that's a negative power of the RG time $t = \log \Lambda$. By the assumption of critical slowing down, this can also be formulated as a factor $x^c$ for some $c > 0$.
    \begin{equation}
        \int_{xl}^{l} \left(- \frac{y}{\Lambda} + \frac{\Lambda \mu^I_{\Lambda}(\Lambda)}{\ev{d}} - \frac{2 C^{TI}_{\Lambda}(l,\Lambda)}{\rho^T_{\Lambda}(l)}\right) d \Lambda \lesssim \int_{\log (xl)}^{\log l} d (\log \Lambda) x^c \approx \log x^{-1} x^c \,.
    \end{equation}
    Along the separatrix, $x \rightarrow 0$ at large $\Lambda$. Therefore, for sufficiently large $\Lambda$, we have
    \begin{equation}
        r_l(l) \approx e^{- \log x^{-1} x^c} r_{\Lambda}(l) \approx (1 - \log x^{-1} x^c) r_{\Lambda}(l) \,.
    \end{equation}
    implying that $r_l(l), r_{xl}(l)$ are equal up to small corrections. Since $r_{\Lambda}(l)$ decreases monotonically from $\Lambda = xl$ to $l$, it must in fact be true that $r_l(l) \approx r_{\Lambda}(l)$ for all $\Lambda \in (xl,l)$ which is what we set out to prove.
    
    Below the separatrix, we can reexamine the scaling of all three terms. As the flow deviates from the separatrix and tends towards the MBL fixed line, $y$ and $\mu^I_{\Lambda}(\Lambda)$ decay to 0 faster, while $\ev{d}$ grows to $\infty$ faster as $\Lambda$ increases, since the decimation rate of T-blocks increases monotonically. This trend makes all three terms decay faster towards 0. Hence we still have $r_l(l) \approx r_{\Lambda}(l)$ for all $\Lambda \in (xl,l)$. For very large $\Lambda$, $\rho^T_{\Lambda}(\Lambda)$ will settle into some asymptotic scaling with $\Lambda$ such that $\rho^T_{\Lambda}(\Lambda) \gg \mathcal{O}(1)\frac{1}{\Lambda})$. This implies that $\ev{d}, y \sim e^{\int^{\Lambda} \rho^T_{\Lambda'}(\Lambda')} \sim \rm Superpoly(\Lambda)$. Hence, the asymptotic error rate is 
    \begin{equation}
        r_l(l) \approx r_{\Lambda}(l) + \mathcal{O}\left( \frac{1}{\rm Superpoly(\Lambda)} \right) \,,
    \end{equation}
    as we set out to show. 
\end{proof}
\begin{lemma}\label{lemma3:appendix}
    The following more concrete estimates about the distribution of T-block lengths hold along the separatrix: $\rho^T_{\Lambda}(l) \approx \frac{x_{\Lambda}^{-1} \Lambda \log \Lambda}{x_l^{-1} l^2 \log l}$ for $l \in [\Lambda, 2\Lambda + \frac{\Lambda}{x}]$. This in turn implies that $\ev{l} \approx \Lambda \log x^{-1}$. 
\end{lemma}
\begin{proof}
    From Lemma~\ref{lemma2:appendix}, we know that for $l \in [\Lambda, \frac{\Lambda}{x} + 2 \Lambda]$,
    \begin{equation}
        r_l(l) \approx r_{\Lambda}(l) \quad r_{\Lambda}(l) = \frac{x_{\Lambda}}{\ev{d}(\Lambda)} \rho^T_{\Lambda}(l) \,.
    \end{equation}
    Using the scaling estimates for $\ev{d}(\Lambda), \rho^T_{\Lambda}(\Lambda)$ and assumption 1, we can immediately prove the estimate for $\rho^T_{\Lambda}(l)$
    \begin{equation}
        \frac{x_l}{l \log l} \frac{1}{l} \approx \frac{x_{\Lambda}}{\Lambda \log \Lambda} \rho^T_{\Lambda}(l) \quad \rightarrow \quad \rho^T_{\Lambda}(l) \approx \frac{x_{\Lambda}^{-1} \Lambda \log \Lambda}{x_l^{-1} l^2 \log l} \,.
    \end{equation}
    The above formula implies an estimate on the average T-block length
    \begin{equation}
        \ev{l} \approx \int_{\Lambda}^{\Lambda/x} \frac{x_{\Lambda}^{-1} \Lambda \log \Lambda}{x_l^{-1} l^2 \log l} l dl \approx x_{\Lambda}^{-1} \Lambda \log \Lambda \int_{\log \Lambda}^{\log \Lambda + \log x^{-1}} \frac{1}{x_l^{-1} \log l} d (\log l)  \,.
    \end{equation}
    By assumption 1 again, $x_{\Lambda}^{-1} = \text{Poly}(\log \Lambda)$. This immediately implies that 
    \begin{equation}
        \ev{l} \approx \frac{x_{\Lambda}^{-1}\Lambda \log \Lambda}{x_{\Lambda}^{-1} \log \Lambda} \log x^{-1} \approx \Lambda \log x^{-1} \,,
    \end{equation}
    which is what we set out to show. 
\end{proof}
Using the estimates for $\rho^T_{\Lambda}(\Lambda)$, we can obtain some more precise control over $\mu^I_{\Lambda}(\Lambda)$ and $x$:
\begin{corollary}\label{cor:appendix}
    Under assumptions 1 and 3 and along the separatrix, $\mu^I_{\Lambda}(\Lambda)/x \approx \rho^T_{\Lambda}(\Lambda) + \mathcal{O}(1)\frac{1}{\Lambda \log \Lambda})$. 
\end{corollary}
\begin{proof}
    By assumptions 1 and 3, the previous two lemmas hold and we have $\mu^I_{\Lambda}(\Lambda) < x \rho^T_{\Lambda}(\Lambda)$. Here, we want to show that they are in fact equal to leading order in $\Lambda$. For that we return to the flow equation for the thermal fraction $f = \frac{\ev{l^T}}{\ev{l^I} + \ev{l^T}}$ (for notational clarity we use $\ev{l^T}$ instead of $\ev{l}$ to denote the average T-block length in this proof). Along the critical separatrix, we have $\ev{l^T} = \Lambda f_T(\Lambda), \ev{l^I} = \Lambda f_I(\Lambda)$ where $f_T(\Lambda) \sim \log x^{-1}$ and $f_I(\Lambda) \sim \frac{\log \Lambda}{x}$ as we flow towards the fixed point at $f = 0$. The flow equation for $f$ dictates that
    \begin{equation}
        \frac{d \log f}{d \log \Lambda} \approx -\frac{\Lambda^2}{\ev{l^T}} \left[-\mu^I_{\Lambda}(\Lambda)/x + \rho^T_{\Lambda}(\Lambda)\right] \,.
    \end{equation}
    Now let us evaluate the LHS approximately
    \begin{equation}
        \frac{d \log f}{d \log \Lambda} = \frac{d \log f_T(\Lambda)}{d \log \Lambda} - \frac{d \log \left[f_T(\Lambda) + f_I(\Lambda)\right]}{d \log \Lambda} \approx - \frac{d \log f_I(\Lambda)}{d \log \Lambda} + \text{subleading} \,,
    \end{equation}
    where in the last step we used the fact that $f_I(\Lambda)$ grows much faster than $f_T(\Lambda)$ at large $\Lambda$. This estimate implies that:
    \begin{equation}
        \frac{d \log f_I(\Lambda)}{d \log \Lambda} \approx \mathcal{O}\left(\frac{\Lambda \left[-\mu^I_{\Lambda}(\Lambda)/x + \rho^T_{\Lambda}(\Lambda)\right]}{\log x^{-1}} \right) \,.
    \end{equation}
    But if $f_I(\Lambda) \sim \log \Lambda/x$, $\frac{d \log f_I(\Lambda)}{d \log \Lambda} = \mathcal{O}(1)\frac{1}{\log \Lambda})$. Since $\log x^{-1}$ is smaller than any polynomial in $\log \Lambda$, the equation above is inconsistent unless $\mu^I_{\Lambda}(\Lambda)/x - \rho^T_{\Lambda}(\Lambda) = \mathcal{O}(1)\frac{1}{\Lambda \log \Lambda})$. 
\end{proof}
These lemmas already allow us to study the relationship between $x,y$ strictly along the separatrix. But since we are interested in the entire region below the separatrix, we need to give a more precise argument for how the RHS of \eq{eq:rlambda_finite_flow_appendix} scales. This is the goal for the rest of this appendix. In step (1), we obtain a compact formula for $F_{\Lambda, \text{prod}}$ valid in the large $\Lambda$ limit. In step (2) we argue that $F_{\Lambda,\text{depl}}$ is suppressed relative to $F_{\Lambda, \text{prod}}$, thereby establishing the recursion relation for $r_{\Lambda}(l)$ and hence for $y = r_{\Lambda}(\Lambda) \Lambda^2$. 

As a preface to the remaining arguments, we emphasize that the logic above is a direct generalization of the logic for uncorrelated MHI explained in Ref.~\onlinecite{Morningstar_Huse_Imbrie_2020}. In the uncorrelated case, step (1) requires a careful analysis of the decay properties of the integrand $\rho^T_{\Lambda'}(l_1) \mu^I_{\Lambda'}(\Lambda') \rho^T_{\Lambda'}(l-\Lambda'/x-l_1)$ in the integration domain Fig.~\ref{fig:IntDomain}. For the correlated case, the integration domain doesn't change, but the integrand no longer factorizes. Thus the main challenge is to use the lemmas and fundamental assumptions to argue for similar decay properties. Step (2) in the uncorrelated case is a direct application of Lemma~\ref{lemma2:appendix}. With correlations, the argument essentially goes through unscathed given the estimates already proven in the lemmas. The recursion relation can then be plugged into the main text to derive critical properties.

Finally, for readers who are already familiar with the MHI argument in Ref.~\onlinecite{Morningstar_Huse_Imbrie_2020}, we will make some occasional comments emphasizing essential ingredients here that are different from those in MHI. Our argument should be comprehensible without reading these comments. 

\subsection{Compact formula for \texorpdfstring{$F_{\Lambda, \text{prod}}$}{}}

We start by writing down the production term without splitting it into disconnected and connected parts:
\begin{equation}
    F_{\Lambda, \rm prod}(l) = \int_{x\Lambda}^{\Lambda} d \Lambda' \frac{x_{\Lambda'}}{\ev{d}_{\Lambda'}} \int_{\Lambda'}^{l - \Lambda'(1+x^{-1})} C^{TIT}_{\Lambda'}\left(l_1,\Lambda', l-\frac{\Lambda'}{x} - l_1\right) dl_1 \,.
\end{equation}
Following Ref.~\onlinecite{Morningstar_Huse_Imbrie_2020}, we make a change of variables from $\Lambda' \rightarrow l_2 = l - \frac{\Lambda'}{x} - l_1$. Since $x$ varies slowly with $\Lambda$, we can disregard the dependence of $x$ on $\Lambda$ and pull it out of the integrals. This allows us to simplify the production term as
\begin{equation}
    F_{\Lambda, \rm prod}(l) \approx x^2 \int_{D} dl_2 dl_1 \frac{1}{\ev{d}_{\Lambda'}} C^{TIT}_{\Lambda'}(l_1, \Lambda', l_2) \,,
\end{equation}
where the integration domain $D$ is the blue isosceles triangle shown in Fig.~\ref{fig:IntDomain}. 
\begin{figure}
    \centering
    \includegraphics[width=\linewidth/2,clip]{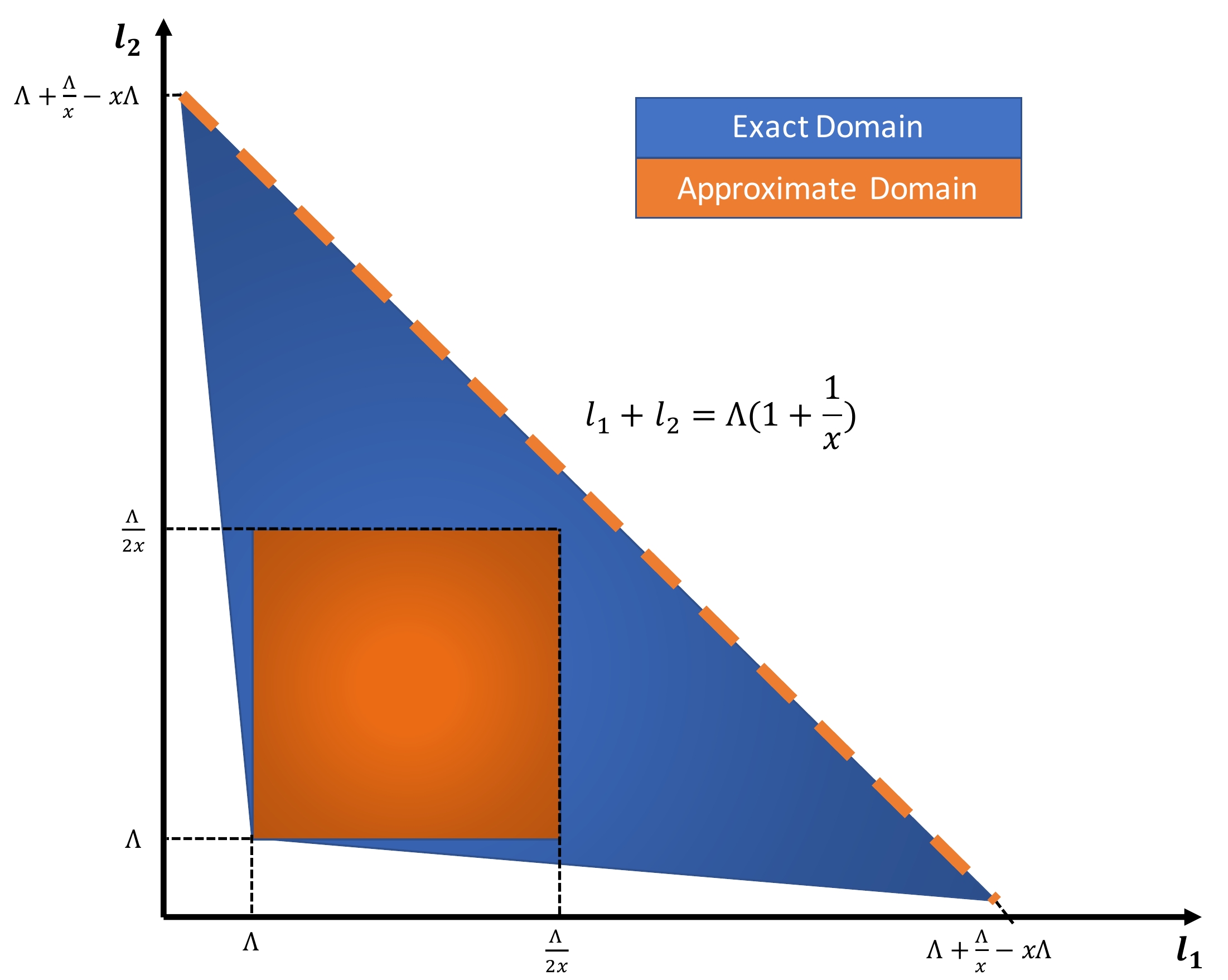}
    \caption{In this plot we show the exact and approximate integration domains relevant for the production term . The opening angle of the isosceles triangle (exact domain) approaches $\pi/2$ as $x \rightarrow 0$. }
    \label{fig:IntDomain}
\end{figure}
The integrand now involves a three-block correlation. We use the flow equation \eq{eq:3_neighbor_corr_flow} derived before to bound its growth in the region $x \max(l_1,l_2)\leq \Lambda' \leq \min(l_1,l_2)$ (this range is chosen so that we can drop the production term in \eq{eq:3_neighbor_corr_flow} involving integrals over five-block distributions):
\begin{equation}
    \begin{aligned}
        \partial_{\Lambda'} (\frac{C^{TIT}_{\Lambda'}(l_1,\Lambda',l_2)}{\ev{d}_{\Lambda'}}) &\approx - \frac{1}{\ev{d}_{\Lambda'}^2} \left[\rho^T_{\Lambda'}(\Lambda') + \mu^I_{\Lambda'}(\Lambda')\right] \ev{d}_{\Lambda'} C^{TIT}_{\Lambda'}(l_1,\Lambda',l_2) \\
        &+ \frac{1}{\ev{d}_{\Lambda'}} \bigg\{C^{TIT}_{\Lambda'}(l_1,\Lambda',l_2)\left[\rho^T_{\Lambda'}(\Lambda') + \mu^I_{\Lambda'}(\Lambda')\right] - C^{TITI}_{\Lambda'}(l_1,\Lambda',l_2,\Lambda') \\
        &- C^{ITIT}_{\Lambda'}(\Lambda',l_1,\Lambda', l_2) + \partial_d C^{TIT}_{\Lambda'}(l_1, d, l_2)|_{d=\Lambda'}\bigg\} \,.
    \end{aligned}
\end{equation}
The first line cancels nicely with the first term on the second line. The remaining terms are:
\begin{equation}
    \partial_{\Lambda'} (\frac{C^{TIT}_{\Lambda'}(l_1,\Lambda',l_2)}{\ev{d}_{\Lambda'}}) \approx \frac{1}{\ev{d}_{\Lambda'}} \bigg\{ - C^{TITI}_{\Lambda'}(l_1,\Lambda',l_2,\Lambda') - C^{ITIT}_{\Lambda'}(\Lambda',l_1,\Lambda', l_2) + \partial_d C^{TIT}_{\Lambda'}(l_1, d, l_2)|_{d=\Lambda'}\bigg\} \,.
\end{equation}
Using assumptions 2 and 3, all three terms can be bounded by $\frac{1}{\ev{d}_{\Lambda'}} C^{TIT}_{\Lambda'}(l_1, \Lambda', l_2) \rho^T_{\Lambda'}(\Lambda') x^c$. By an easy adaptation of the argument in Lemma~\ref{lemma2:appendix}, we have the uniform estimate:
\begin{equation}
    \boxed{\frac{C^{TIT}_{\Lambda'}(l_1,\Lambda',l_2)}{\ev{d}_{\Lambda'}} \approx \text{const} \quad \forall \max(l_1,l_2)x \leq \Lambda' \leq \min(l_1,l_2) \,.}
\end{equation}
Now we are ready to tackle the double integral over $D$ (the blue + orange region in Fig.~\ref{fig:IntDomain}). The subset of the domain where $l_1,l_2 \geq \Lambda$ is an isosceles right triangle with leg length $\frac{\Lambda}{x} - x\Lambda$. Denote this right triangle by $T$ and the remaining narrow wedges $D \setminus T$. Since the two narrow wedges are symmetric about the line $l_1 = l_2$, we need only demonstrate that the integral over one of the wedges is suppressed relative to the integral over the right triangular domain. This can be accomplished in two steps:
\begin{enumerate}
    \item Let us consider first the right triangular domain $T$. The difficulty in estimating this integral is that the cutoff $\Lambda'$ of the integrand depends implicitly on $l_1, l_2$. We would like to make a series of approximations until we can evaluate the integral at a fixed cutoff. 
    
    To do that, we first consider the orange square $[\Lambda, \frac{\Lambda}{2x}] \times [\Lambda, \frac{\Lambda}{2x}]$ contained in $T$. Outside of this square, $l_1 \geq \frac{\Lambda}{2x}$ or $l_2 \geq \frac{\Lambda}{2x}$, implying that $\Lambda' = \Lambda + x (2\Lambda - l_1-l_2) \leq \frac{\Lambda}{2}$. As a result, $\frac{\Lambda'}{l_1} \leq x$ or $\frac{\Lambda'}{l_2} \leq x$. But since the least unlikely way to produce a large T-block with $l_1 \geq \frac{\Lambda'}{x}$ is to combine two positively correlated T-blocks at cutoff, $C^{TIT}_{\Lambda'}(l_1, \Lambda', l_2)$ must be exponentially decaying in $l_1,l_2$ for all $(l_1,l_2) \notin [\Lambda, \frac{\Lambda}{2x}] \times [\Lambda, \frac{\Lambda}{2x}]$ with a decay length on the order of $\frac{\Lambda'}{x}$. These exponential tails will give small corrections and for a leading order approximation we can restrict the integral to the orange square from now on.
    
    For readers familiar with MHI, note that our argument is \textit{subtly different} from the original MHI argument that replaced the cutoff $\Lambda'$ with $\Lambda$ everywhere in the orange square. This replacement is puzzling because $x \max (l_1,l_2) \leq \Lambda' \leq \min (l_1, l_2)$ is \textbf{not true everywhere in D} and the analogue of Lemma~\ref{lemma2:appendix} for uncorrelated RG cannot be applied. We circumvent this possible loophole by restricting to an even smaller square $S' = [\Lambda, \frac{\Lambda}{x^{1/2}}] \times [\Lambda, \frac{\Lambda}{x^{1/2}}]$. In this smaller domain, $x\max(l_1,l_2) = x^{1/2} \Lambda, \min(l_1,l_2) = \Lambda$ and $\Lambda' = \Lambda + 2x (\Lambda - l_1 - l_2) \leq \Lambda + 2x \Lambda - 2 x^{1/2} \Lambda$. When $x \ll 1$, $x \max(l_1,l_2) \leq \Lambda' \leq \min(l_1,l_2)$ indeed holds and
    \begin{equation}
        \frac{C^{TIT}_{\Lambda'}(l_1, \Lambda', l_2)}{\ev{d}_{\Lambda'}} \approx \frac{C^{TIT}_{\Lambda}(l_1, \Lambda, l_2)}{\ev{d}_{\Lambda}} \quad \forall (l_1,l_2) \in S' \,.
    \end{equation}
    Within $S'$, we can also estimate the decay of the integrand with $l_1,l_2$. Without loss of generality, let $l_1 \leq l_2$. By monotonicity of the probability distributions, 
    \begin{equation}
        \frac{C^{TIT}_{\Lambda'}(l_1, \Lambda', l_2)}{\ev{d}_{\Lambda'}}  \approx \frac{C^{TIT}_{\Lambda}(l_1, \Lambda, l_2)}{\ev{d}_{\Lambda}} \quad \forall l_1, l_2 \in [\Lambda, \frac{\Lambda}{x^{1/2}}] \,.
    \end{equation}
    For positive correlations, at cutoff $l_1$, $\ev{d}, \ev{l} \gg l_1,l_2$ and $C^{TIT}_{\Lambda}(l_1, \Lambda, l_2)$ should be enhanced relative to $\rho^T_{\Lambda}(l_1)\rho^T_{\Lambda}(l_2) \mu^I_{\Lambda}(\Lambda)$. However, the enhancement does not change the power law scaling in $l_1, l_2$ according to assumption 3. Therefore, by Lemma~\ref{lemma3:appendix}, we have that along the separatrix,
    \begin{equation}
        \frac{C^{TIT}_{\Lambda'}(l_1, \Lambda', l_2)}{\ev{d}_{\Lambda'}} \sim \mathcal{O}\left(\frac{1}{l_1^2 l_2^2}\right)  \,.
    \end{equation}
    When $(l_1,l_2) \in S \setminus S'$, we necessarily have $\frac{C^{TIT}_{\Lambda'}(l_1, \Lambda', l_2)}{\ev{d}_{\Lambda'}} \leq \frac{C^{TIT}_{\Lambda}(l_1, \Lambda, l_2)}{\ev{d}_{\Lambda}}$ since $\Lambda' < \Lambda$. The $\frac{1}{l_1^2l_2^2}$ decay therefore guarantees that the contributions from the region $S \setminus S'$ is suppressed by a factor of $x$ relative to the contributions from $S'$ and can be thus neglected. Below the separatrix, the decay of $C^{TIT}_{\Lambda}(l_1, \Lambda, l_2)$ with $l_1,l_2$ has to be faster because flowing towards the MBL fixed line implies a clustering of T-blocks close to the cutoff. Therefore, the suppression factor is much smaller than $x$ and there is no additional complication.
    
    After restricting to $S$, we can shift the cutoff uniformly to $\Lambda$ so that
    \begin{equation}
        x^2 \int_{S'} \frac{C^{TIT}_{\Lambda'}(l_1, \Lambda', l_2)}{\ev{d}_{\Lambda'}} \approx x^2 \int_{S} \frac{C^{TIT}_{\Lambda}(l_1, \Lambda, l_2)}{\ev{d}_{\Lambda}} \,.
    \end{equation}
    Now up to errors that are suppressed by positive powers of $x$, we can extend the integration domain to $[\Lambda, \infty]^2$. Therefore, to leading order,
    \begin{equation}
         x^2 \int_{S} \frac{C^{TIT}_{\Lambda}(l_1, \Lambda, l_2)}{\ev{d}_{\Lambda}} = x^2 \int_{[\Lambda, \infty]^2} \frac{1}{\ev{d}_{\Lambda}}C^{TIT}_{\Lambda}(l_1, \Lambda, l_2) = \frac{x^2 \mu^I_{\Lambda}(\Lambda)}{\ev{d}_{\Lambda}} \,.
    \end{equation}
    where the last identity follows from the definition of $C^{TIT}_{\Lambda}$.  
    \item The integral over the narrow wedges is much easier to analyze. As explained above, the value of $C^{TIT}_{\Lambda'}(l_1, \Lambda', l_2)$ is only appreciable near $(l_1, l_2) = (\Lambda, \Lambda)$ and exponentially suppressed for $l_1, l_2 \geq \frac{\Lambda}{2x}$. But the area of the wedges restricted to $l_1, l_2 \leq \frac{\Lambda}{2x}$ is only $\mathcal{O}(x)$. Hence relative to the integral in the square region, the wedge integral must be suppressed at least by an additional factor of $x$. We thus arrive at equation \eq{eq:F_Lambda_prod} in the main text:
    \begin{equation}
        F_{\Lambda, \rm prod} \approx \frac{x^2 \mu^I_{\Lambda}(\Lambda)}{\ev{d}_{\Lambda}} \,.
    \end{equation}
\end{enumerate}

\subsection{Bounding \texorpdfstring{$F_{\Lambda, \rm depl}$}{} relative to \texorpdfstring{$F_{\Lambda, \rm prod}$}{}} 

At this point, it is easy to see that $F_{\Lambda, \rm depl}$ is suppressed relative to $F_{\Lambda, \rm prod}$. Recall that since the least unlikely way to create a T-block late in the RG is by combining two T-blocks of length $l^T \approx \Lambda$ with an insulating block with $l^I = \frac{\Lambda}{x}$, we expect $\rho^T_{x\Lambda}(l)$ to decay exponentially for $l> \mathcal{O}(\Lambda)$ (in fact we do not need it to be an exponential. A fast power law is enough). This means that
\begin{equation}
    \frac{r_{\Lambda}(l)}{r_{x\Lambda}(l)} \approx \frac{\ev{d}_{\Lambda} \rho^T_{\Lambda}(l)}{\ev{d}_{x\Lambda} \rho^T_{x\Lambda}(l)} \approx \mathcal{O}(x^{-1} e^{1/x}) \gg 1 \,.
\end{equation}
Since $x \ll 1$, this immediately shows that $F_{\Lambda, \text{prod}}$ should dominate over $F_{\Lambda, \text{depl}}$. Hence we obtain the recursion relation
\begin{equation}
    r_{\Lambda}(l) = r_{\Lambda}(\Lambda/x) \approx \frac{x^2 \mu^I_{\Lambda}(\Lambda)}{\ev{d}} \,.
\end{equation}
Now using Lemma~\ref{lemma2:appendix}, we can turn the above equation into a recursion relation for $y$:
\begin{equation}
    y_{\Lambda/x} = \frac{\Lambda^2}{x^2} r_{\Lambda/x}(\Lambda/x) \approx  \frac{\Lambda^2}{x^2} r_{\Lambda}(\Lambda/x) \approx \frac{\Lambda^2}{x^2} \frac{x^2 \mu^I_{\Lambda}}{\ev{d}} \approx \frac{\Lambda^2}{\ev{d}^2} \ev{d} \mu^I_{\Lambda}(\Lambda) = \left(\frac{y_{\Lambda}}{x}\right)^2 \ev{d} \mu^I_{\Lambda}(\Lambda) \,.
\end{equation}
This recursion is the same as the uncorrelated recursion up to a factor $\ev{d} \mu^I_{\Lambda}(\Lambda)$.

\section{Effects of long range correlated disorder on the symmetric RG}\label{sec:AppendixC}

In this appendix, we expand upon a comment in Sec.~\ref{sec:stability_hyp} about the effect of spatial correlations on earlier iterations of MBL RGs~\cite{Zhang_Zhao_Devakul_Huse_2016, Goremykina_Vasseur_Serbyn_2019}. The RG rules in these models can be summarized as 
\begin{equation}
    l^I_{\rm new} = l^I_{i-1} + \alpha_T l^T_i (=\Lambda) + l^I_{i+1}, 
\end{equation}
\begin{equation}
    l^T_{\rm new} = l^T_{i-1} + \alpha_I l^I_i (=\Lambda) + l^T_{i+1},
\end{equation}
where $\alpha_I, \alpha_T$ are tunable parameters satisfying $\alpha = \alpha_I = \alpha_T^{-1}$. In the case of $\alpha = 1$ we recover the symmetric RG of Ref.~\onlinecite{Zhang_Zhao_Devakul_Huse_2016}.

Now we choose the initial block length distributions so that $\ev{l^T} = W$ and $\ev{l^I} = 1$. For a particular $\alpha$ and spatial correlation, there is a critical $W_c(\alpha)$ where the production of T/I-blocks exchange dominance. The critical exponent $\nu(\alpha)$ of the phase transition at $W = W_c(\alpha)$ can be extracted from a finite size scaling analysis. From here on we focus on the symmetric RG with $\alpha_I = \alpha_T = W_c = 1$.  

When spatial correlations are present, we expect the critical exponent to deviate from the value $\nu = 2.5$ at the uncorrelated fixed point. When we turn on positive correlations, numerics show a clear upward drift in $\nu$ as $c \rightarrow 0$ (remember that $c$ is the decay exponent of the initial correlations and small c corresponds to strong correlations). The numerical values for $\nu$ satisfy the generalized Harris bound, although a precise extrapolation of $\nu(L_{\rm max} \rightarrow \infty)$ is not possible due to finite size effects (for concreteness, note the drifts shown in Fig.~\ref{fig:symmetricRG_nu}). For hyperuniform correlations, the numerically extracted values of $\nu$ stay close to $\nu = 2.5$ for all $\alpha$, consistent with the general argument in Sec.~\ref{sec:stability_hyp} that hyperuniform correlations should be irrelevant for asymptotically additive RG flows. At the special point at $\alpha = 1$, the wandering exponent $w = 0$ coincides with the wandering exponent of quasi-periodic correlations. Interestingly, the symmetric RG with quasi-periodic correlations have been shown to have $\nu = 1$ rather than $\nu = 2.5$ (see Ref.~\onlinecite{Agrawal_Gopalakrishnan_Vasseur_2020_QP}). This means that as far as the symmetric RG is concerned, random disorder with $w = 0$ is qualitatively different from quasiperiodic disorder. Whether or not this qualitative difference exists for the MHI RG is an interesting question to address in the future. 
\begin{figure}
    \centering
    \includegraphics[width = \textwidth/2-9pt]{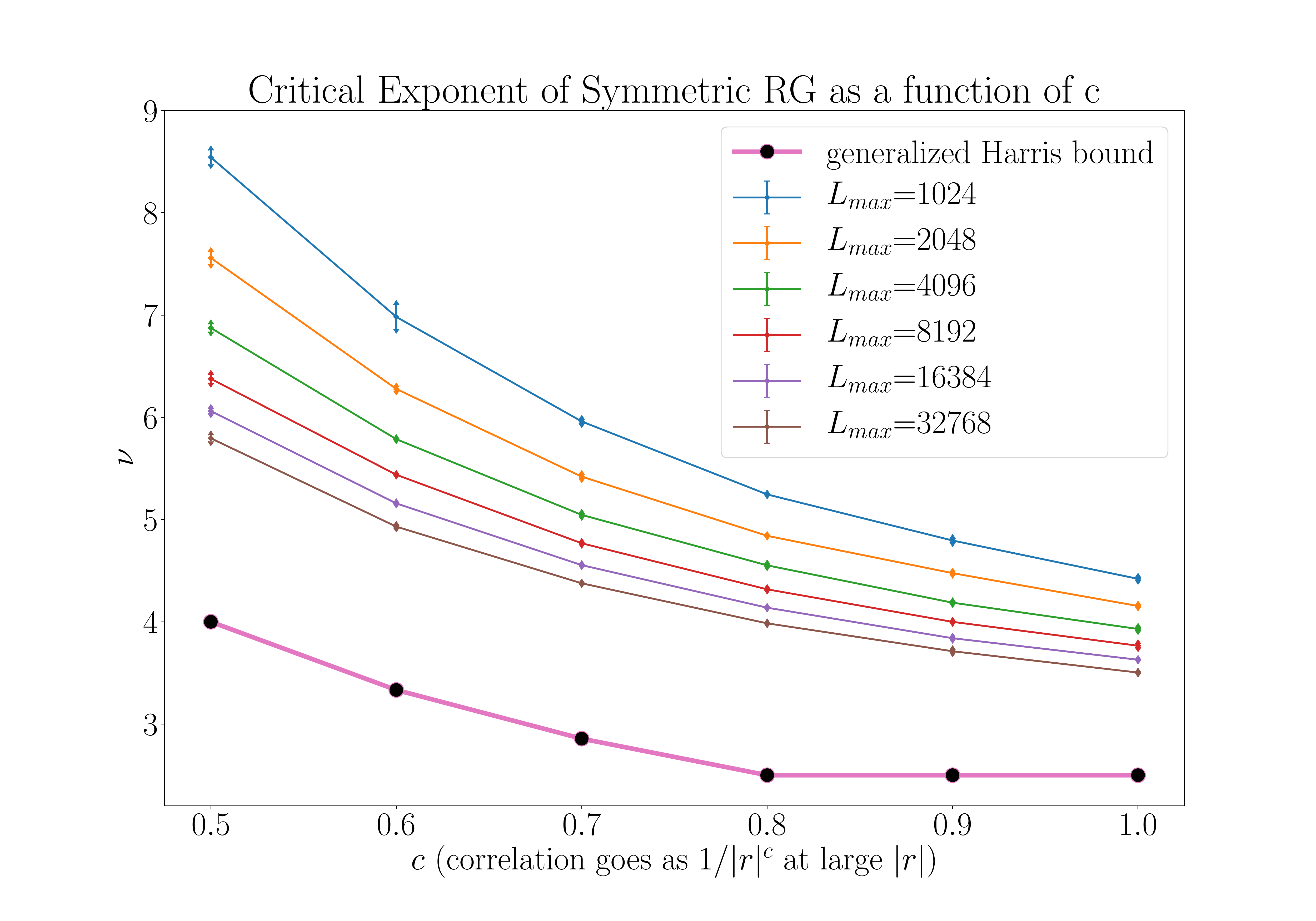}
    \includegraphics[width = \textwidth/2-9pt]{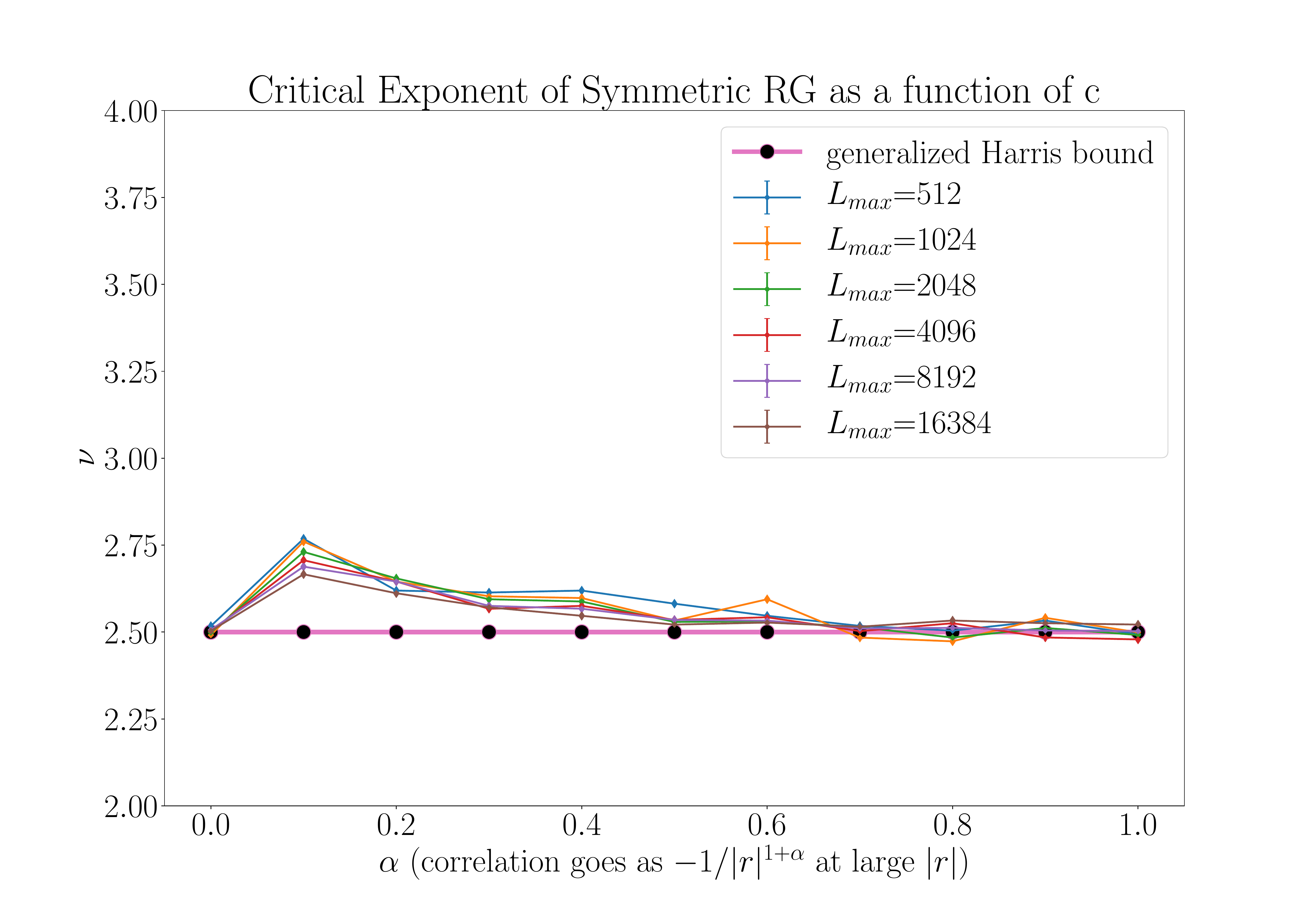}
    \caption{The left/right panels show numerical estimates of $\nu$ for positive and hyperuniform correlations as a function of correlation strength as measured by $c$ and maximum system size $L_{\rm max}$ used in the scaling collapse. For positive correlations, the critical exponent $\nu$ drifts down as $L_{\rm max}$ increases. In the accessible range of $L_{\rm max}$, we are not able to confirm whether the generalized Harris bound is saturated. For hyperuniform correlations, the finite size effects are weaker since coherent fluctuations of large regions are suppressed. For $0.7 \lesssim \alpha$, we confirm that hyperuniform correlations do not change the uncorrelated value of $\nu$, consistent with the general arguments in Sec.~\ref{sec:stability_hyp}. For $\alpha < 0.7$, the numerical $\nu$ slightly exceeds the uncorrelated $\nu = 2.5$. This deviation is at the 10\% level and much smaller than the deviation for positive correlations at any value of $c$. Thus we tentatively attribute it to finite size effects.}
    \label{fig:symmetricRG_nu}
\end{figure}

\end{widetext}

\bibliography{MBLRG_submission.bib}
\bibliographystyle{apsrev4-2}

\end{document}